\documentclass[10pt,pra,aps,nopacs,onecolumn,twoside,superscriptaddress]{revtex4}

\usepackage{amsmath,amsfonts,amssymb,color,graphics,graphicx,latexsym,revsymb,amsthm,url,verbatim,appendix,epstopdf}
\usepackage{algorithm}

\usepackage{algcompatible}

\usepackage{hyperref}

\usepackage{tikz}
\usetikzlibrary{decorations.shapes}
\usetikzlibrary{shapes.symbols}
\usetikzlibrary{decorations.pathmorphing}

\input{qcircuit}

\floatname{algorithm}{Scheme}

\newtheorem{theorem}{Theorem}
\newtheorem{corollary}[theorem]{Corollary}
\newtheorem{lemma}[theorem]{Lemma}

\theoremstyle{definition}
\newtheorem{definition}{Definition}
\newtheorem{proposition}[definition]{Proposition}

\newcommand{\meas}{
\begin{tikzpicture}
\filldraw[fill=white] (0,.25) rectangle (.7,-.25);
\draw (.67,-.1) arc (50:130:.5);
\draw (.35,-.2)--(.525,.2);
\end{tikzpicture}
}

\newcommand{\cntrl}{
\begin{tikzpicture}
\fill (0,0) circle (.08);
\end{tikzpicture}
}

\def\squareforqed{\hbox{\rlap{$\sqcap$}$\sqcup$}}
\def\qed{\ifmmode\squareforqed\else{\unskip\nobreak\hfil
\penalty50\hskip1em\null\nobreak\hfil\squareforqed
\parfillskip=0pt\finalhyphendemerits=0\endgraf}\fi}
\def\endenv{\ifmmode\;\else{\unskip\nobreak\hfil
\penalty50\hskip1em\null\nobreak\hfil\;
\parfillskip=0pt\finalhyphendemerits=0\endgraf}\fi}
\def\Dbar{\leavevmode\lower.6ex\hbox to 0pt
{\hskip-.23ex\accent"16\hss}D}

\makeatletter
\def\bcj{\begin{conjecture}}
\def\ecj{\end{conjecture}}
\def\bcr{\begin{corollary}}
\def\ecr{\end{corollary}}
\def\bd{\begin{definition}}
\def\ed{\end{definition}}
\def\bea{\begin{eqnarray}}
\def\eea{\end{eqnarray}}
\def\bem{\begin{enumerate}}
\def\eem{\end{enumerate}}
\def\bex{\begin{example}}
\def\eex{\end{example}}
\def\bim{\begin{itemize}}
\def\eim{\end{itemize}}
\def\bl{\begin{lemma}}
\def\el{\end{lemma}}
\def\bpf{\begin{proof}}
\def\epf{\end{proof}}
\def\bpp{\begin{proposition}}
\def\epp{\end{proposition}}
\def\bqu{\begin{question}}
\def\equ{\end{question}}
\def\br{\begin{remark}}
\def\er{\end{remark}}
\def\bt{\begin{theorem}}
\def\et{\end{theorem}}

\def\btb{\begin{tabular}}
\def\etb{\end{tabular}}

\newcommand{\nc}{\newcommand}

\def\e{\epsilon}

 \nc{\bA}{{\bf A}} \nc{\bB}{{\bf B}} \nc{\bC}{{\bf C}}
 \nc{\bD}{{\bf D}} \nc{\bE}{{\bf E}} \nc{\bF}{{\bf F}}
 \nc{\bG}{{\bf G}} \nc{\bH}{{\bf H}} \nc{\bI}{{\bf I}}
 \nc{\bJ}{{\bf J}} \nc{\bK}{{\bf K}} \nc{\bL}{{\bf L}}
 \nc{\bM}{{\bf M}} \nc{\bN}{{\bf N}} \nc{\bO}{{\bf O}}
 \nc{\bP}{{\bf P}} \nc{\bQ}{{\bf Q}} \nc{\bR}{{\bf R}}
 \nc{\bS}{{\bf S}} \nc{\bT}{{\bf T}} \nc{\bU}{{\bf U}}
 \nc{\bV}{{\bf V}} \nc{\bW}{{\bf W}} \nc{\bX}{{\bf X}}
 \nc{\bZ}{{\bf Z}}

\nc{\cA}{{\cal A}} \nc{\cB}{{\cal B}} \nc{\cC}{{\cal C}}
\nc{\cD}{{\cal D}} \nc{\cE}{{\cal E}} \nc{\cF}{{\cal F}}
\nc{\cG}{{\cal G}} \nc{\cH}{{\cal H}} \nc{\cI}{{\cal I}}
\nc{\cJ}{{\cal J}} \nc{\cK}{{\cal K}} \nc{\cL}{{\cal L}}
\nc{\cM}{{\cal M}} \nc{\cN}{{\cal N}} \nc{\cO}{{\cal O}}
\nc{\cP}{{\cal P}} \nc{\cQ}{{\cal Q}} \nc{\cR}{{\cal R}}
\nc{\cS}{{\cal S}} \nc{\cT}{{\cal T}} \nc{\cU}{{\cal U}}
\nc{\cV}{{\cal V}} \nc{\cW}{{\cal W}} \nc{\cX}{{\cal X}}
\nc{\cZ}{{\cal Z}}

\nc{\hA}{{\hat{A}}} \nc{\hB}{{\hat{B}}} \nc{\hC}{{\hat{C}}}
\nc{\hD}{{\hat{D}}} \nc{\hE}{{\hat{E}}} \nc{\hF}{{\hat{F}}}
\nc{\hG}{{\hat{G}}} \nc{\hH}{{\hat{H}}} \nc{\hI}{{\hat{I}}}
\nc{\hJ}{{\hat{J}}} \nc{\hK}{{\hat{K}}} \nc{\hL}{{\hat{L}}}
\nc{\hM}{{\hat{M}}} \nc{\hN}{{\hat{N}}} \nc{\hO}{{\hat{O}}}
\nc{\hP}{{\hat{P}}} \nc{\hR}{{\hat{R}}} \nc{\hS}{{\hat{S}}}
\nc{\hT}{{\hat{T}}} \nc{\hU}{{\hat{U}}} \nc{\hV}{{\hat{V}}}
\nc{\hW}{{\hat{W}}} \nc{\hX}{{\hat{X}}} \nc{\hZ}{{\hat{Z}}}

\nc{\hn}{{\hat{n}}}

\def\tr{\mathop{\rm Tr}}

\def\ox{\otimes}

\newcommand{\ketbra}[2]{|#1\rangle\!\langle#2|}

\newcommand{\tgate}{{\sf T}}
\newcommand{\hgate}{{\sf H}}
\newcommand{\pgate}{{\sf P}}
\newcommand{\igate}{{\mathbb I}}
\newcommand{\cnot}{{\sf CNOT}}
\newcommand{\R}{{\sf R}}
\newcommand{\Xgate}{{\sf X}}
\newcommand{\Zgate}{{\sf Z}}

\begin{document}
\title{A quantum homomorphic encryption scheme for polynomial-sized circuits}
\author{Li Yu}\email{yupapers@sina.com}
\affiliation{Department of Physics, Hangzhou Normal University, Hangzhou, Zhejiang 311121, China}
\date{\today}

\begin{abstract}
Quantum homomorphic encryption (QHE) is an encryption method that allows quantum computation to be performed on one party's private data with the program provided by another party, without revealing much information about the data nor about the program to the opposite party. It is known that information-theoretically-secure QHE for circuits of unrestricted size would require exponential resources, and efficient computationally-secure QHE schemes for polynomial-sized quantum circuits have been constructed. In this paper we first propose a QHE scheme for a type of circuits of polynomial depth, based on the rebit quantum computation formalism. The scheme keeps the restricted type of data perfectly secure. We then propose a QHE scheme for a larger class of polynomial-depth quantum circuits, which has partial data privacy. Both schemes have good circuit privacy. We also propose an interactive QHE scheme with asymptotic data privacy, however, the circuit privacy is not good, in the sense that the party who provides the data could cheat and learn about the circuit. We show that such cheating would generally affect the correctness of the evaluation or cause deviation from the protocol. Hence the cheating can be caught by the opposite party in an interactive scheme with embedded verifications. Such scheme with verification has a minor drawback in data privacy. Finally, we show some methods which achieve some nontrivial level of data privacy and circuit privacy without resorting to allowing early terminations, in both the QHE problem and in secure evaluation of classical functions. The entanglement and classical communication costs in these schemes are polynomial in the circuit size and the security parameter (if any).
\end{abstract}
\maketitle

%\tableofcontents

\section{Introduction}\label{sec1}

As quantum computers are difficult to build, a practical way for doing quantum computation is using the client-server model. The quantum computer is placed at the server, and the clients may have limited quantum capabilities. The client and the server may want to hide some information about their program or data from each other. This is part of the reason why security in bipartite quantum computing is of interest. Bipartite quantum computation may also be related to or directly implement some classical or quantum cryptographic tasks. Two problems in bipartite quantum computing have attracted a lot of attention. One of them is blind quantum computing, and the other is quantum homomorphic encryption, the main topic of this paper.

In classical cryptography, homomorphic encryption (HE) is an encryption scheme that allows computation to be performed (effectively on the plaintext after decryption) while having access only to the ciphertext. A \emph{fully} homomorphic encryption (FHE) scheme allows for \emph{any} computation to be performed in such fashion. The known schemes \cite{Gentry09,brakerski2011efficient} are based on computational assumptions. Quantum (fully) homomorphic encryption (QHE and QFHE respectively) allows (universal) quantum computation to be performed without accessing the unencrypted data. Schemes for quantum homomorphic encryption \cite{rfg12,Fisher13,Tan16,Ouyang18,bj15,Dulek16,Lai17,Mahadev17,TOR18} allow two stages of communication, with some initial shared entangled state. The requirements are that the final computation result is correct, and the data and the final computation result are known only to the data-provider, who learns little about the circuit performed beyond what can be deduced from the result of the computation itself.  There is no limit to the quantum capabilities of any party. We regard interactive QHE schemes as those schemes with similar security characteristics as the usual QHE schemes, but allows many rounds of communication. We allow any party to abort during the interactive protocol.

It is known that information-theoretically-secure (ITS) QHE for circuits of unrestricted size necessarily incurs exponential overhead \cite{NS17,Lai17,Newman18}. For the special case of unrestricted quantum input states and perfect data privacy, a stronger lower bound is stated in \cite{ypf14}. Some computationally-secure QHE schemes for polynomial-sized quantum circuits have been constructed \cite{bj15,Dulek16,Mahadev17}.

In this paper, we firstly introduce a QHE scheme for a restricted type of circuits (Scheme~\ref{sch1}). It keeps the real product input state perfectly secure, and the circuit privacy is optimal up to a constant factor. We then propose a QHE scheme for a larger class of polynomial-depth quantum circuits, which has partial data privacy and near-optimal circuit privacy. The entanglement and classical communication costs in both schemes scale linearly with the product of input size and circuit depth.

We then introduce a partially-secure interactive QHE scheme for polynomial-sized circuits, which is Scheme~\ref{schmain} in Sec.~\ref{sec5}. It is independent from the two schemes above. It uses a subprocedure (Scheme~\ref{schsub}) for computing classical linear polynomials. The Scheme~\ref{schmain} has asymptotic data privacy, but the circuit privacy is not good. We then propose an enhanced version with embedded verifications, in which Bob, who provides the circuit, may choose to abort if the messages returned from Alice during the protocol do not pass his verifications. Such enhanced scheme has both quite satisfactory data privacy and circuit privacy: it leaks a constant amount of information about the data, where the constant is determined by the security requirements; and the circuit privacy is quite good since the cheating by Alice is likely to be caught.

In Sec.~\ref{sec6}, we propose a variant of Scheme~\ref{schsub}, denoted as Scheme~\ref{schsub2}, which achieves some nontrivial level of data privacy and circuit privacy. The degree of data privacy does not satisfy the IND-CPA criterion, but for some range of parameters, it is good when the input is \emph{a priori} uniformly random in the other party's view. There are also other variants which interpolates between Scheme~\ref{schsub} and Scheme~\ref{schsub2}. In Sec.~\ref{sec7}, we propose Scheme~\ref{schsub3} for computing classical linear polynomials. It is based on the idea of information locking by quantum techniques, but also contains some idea from Scheme~\ref{schsub2}. It generally achieves better security than Scheme~\ref{schsub2}. In Sec.~\ref{sec8}, we present methods for combined use of the schemes for evaluating classical linear polynomials, and mention some application to interactive QHE with limited circuit depth. In Sec.~\ref{sec9}, we present some interactive protocols that make use of the previously proposed schemes for evaluating general classical functions with some nontrivial level of security.

The rest of the paper is organized as follows. In Sec.~\ref{sec2} we review some literature results. In Sec.~\ref{sec3} we will introduce some preliminary knowledge. Secs.~\ref{sec4} presents the two QHE schemes which are secure for some types of data and some types of circuits. Sec.~\ref{sec5} contains an interactive QHE scheme and its enhanced version. The Secs.~\ref{sec6} and \ref{sec7} contain variants of Scheme~\ref{schsub}. Sec.~\ref{sec8} contains methods for combined use of the schemes for evaluating classical linear polynomials. Sec.~\ref{sec9} contains some methods for evaluating general classical functions with some nontrivial level of security. Sec.~\ref{sec10} contains the conclusion and some open problems.

\section{Review of past results}\label{sec2}

In the following we put our work in the context of literature results on quantum homomorphic encryption and related problems.

Blind quantum computing (BQC) is a task in which Alice knows both the data and program but can only do limited quantum operations, and the main part of the computation is to be carried out by Bob, who learns neither the data nor much about the program. The output is to be known to Alice but not Bob (this is different from QHE, in which the party with the program does not know the output). Broadbent, Fitzsimons and Kashefi presented a BQC scheme \cite{bfk09} using the measurement-based quantum computing model. An experimental demonstration is in \cite{Barz12}. There are other BQC schemes based on the circuit model \cite{ABE10}, or the ancilla-driven model \cite{skm13}. The possibility of Alice being a classical client (allowing partial leakage of the client's information) is discussed in \cite{Mantri17}.

In Yu \emph{et al} \cite{ypf14}, it is shown that the condition of perfect data privacy implies an information localization phenomenon, and this puts restrictions on the possible forms of QHE schemes and their complexity. The Theorem 1 and Corollary 1 of \cite{ypf14} are for the case that input states are unrestricted. For restricted input states such as real product states, the Scheme~\ref{sch1} in the current work and the scheme in \cite{Lai17} provide examples for which the above statements do not apply.

Newman and Shi \cite{NS17} (see also \cite{Newman18}), concurrently with Lai and Chung \cite{Lai17}, showed that ITS-QFHE schemes require communication cost (including entanglement cost) that is at least exponential in the number of input qubits, where ``ITS'' allows a negligible amount of data leakage. This is based on the work of Nayak \cite{Nayak99}.

On computing two-party classical functions with quantum circuits, Lo \cite{Lo97} studied the data privacy in the case that the output is on one party only, and both parties know the function. The security of two-party quantum computation for fixed classical function with both parties knowing the outcome has been studied by Buhrman \emph{et al} \cite{bcs12}.

There is a line of work on partially-ITS QHE, which are sometimes limited to certain types of circuits. In the scheme of Rohde \emph{et al} \cite{rfg12}, only a logarithmic amount of information about the data is hidden, and the type of circuit is limited to Boson sampling. Tan {\it et al} \cite{Tan16,TOR18} presented QHE schemes that allow some types of quantum circuits and hides only partial information (but more than that in \cite{rfg12}) about the data. Ouyang {\it et al} \cite{Ouyang18} presented a scheme that hides almost all information about the data, but can only perform Clifford gates and a constant number of non-Clifford gates. In the scheme in \cite{Ouyang18}, the party with initial data could cheat by replacing the mixed qubits by pure qubits in some particular input state to gain partial knowledge about the Clifford gates applied, without affecting the correctness of the evaluation. The degree of circuit privacy is related to the size of the code, thus it is in a tradeoff with data privacy.

The following results are based on the existence of computationally-secure classical homomorphic encryption schemes. Broadbent and Jeffery \cite{bj15} presented two computationally-secure QFHE schemes that are relatively efficient for polynomial-sized quantum circuits with limited use of $\tgate$ gates (and arbitrary use of Clifford gates). A compact QFHE scheme was proposed by Dulek, Schaffner and Speelman \cite{Dulek16}. Its complexity of decryption (but not the overall computation and communication costs) could be independent of the size of the circuit to be evaluated, by choosing the classical HE scheme in it appropriately. Mahadev \cite{Mahadev17} presented a quantum leveled fully homomorphic encryption scheme using classical homomorphic encryption schemes with certain properties. The client holding the initial data in the scheme is almost classical except for the sending and receiving of quantum states at the beginning and at the end. Alagic \emph{et al} devised a computationally-secure QFHE scheme with verification \cite{ADSS17}, where the client is similarly almost classical (and completely classical in the case of classical input and output).

Verification of quantum computations is an interesting problem on its own, but the study of this problem has been closely related to that of blind quantum computing. The fact that it is related to QHE has been hinted in \cite{Mahadev17,Mahadev18v,ADSS17}. Usually, the verification problem involves a quantum prover and a quantum verifier, and the prover tries to convince the verifier that a quantum computation has been done as planned. Most studies assume the verifier to have some quantum capabilities \cite{FK17,ABE10,ABOEM17,Broadbent18,FHM18}. There are also significant progress in using a classical or almost classical verifier \cite{ADSS17,DOF18,Mahadev18v,FHM18}. The protocols in \cite{FHM18} show that verification can be done after the computation is done, while not necessarily achieving blindness of the computation.

Some schemes for delegated quantum computation have been proposed \cite{Ch05,Fisher13,MJS16}. In their original form, they usually have good data privacy but are not necessarily strong in circuit privacy. The Supplementary Note 1 of \cite{Fisher13} contains a method based on so-called ``universal circuits'' to achieve blind quantum computing (hiding the information about the circuit which is initially on the same party as the data) but this does not make it useful for the QHE problem, since the input data and the program are on opposite parties in QHE, unlike in BQC. It is pointed out in \cite{bfk09} (with details in \cite{Fitz17}) that the protocol in \cite{Ch05} can be modified to become a blind quantum computing protocol having circuit privacy. This method is also not useful for QHE.

Some lower bounds on the length of classical encryption key in an ITS QHE scheme can be found in \cite{Lai19}.

\section{Preliminaries}\label{sec3}

First, we give a definition for quantum homomorphic encryption. Note the Encryption algorithm below absorbs both the key generation algorithm and the encryption algorithm in \cite{bj15}. The measurements are explicitly allowed.

\begin{definition} \label{def1}
A quantum homomorphic encryption scheme (QHE) $\cE$ contains the following three algorithms, where $\lambda$ is the security parameter, and $T$ refers to a quantum circuit which may contain measurements and quantum gates controlled by classical functions of measurement outcomes,  $\phi$ is an initial entangled resource state shared by Alice and Bob, and $\sigma$ is the input data state of Alice.
\begin{enumerate}
\item Encryption algorithm $Encrypt_\cE(\lambda,\sigma,\phi)$. The algorithm outputs some quantum state $\Pi$ called the \emph{ciphertext}. The reduced state of $\Pi$ on Alice's and Bob's systems are denoted $\Pi_a$ and $\Pi_b$, respectively.
\item Evaluation algorithm $Evaluate_\cE(\lambda,T,\Pi_b)$ that does a computation $T$ on $\Pi_b$ without decryption.
\item Decryption algorithm $Decrypt_\cE(\lambda,\Pi_a,\Pi_b')$, where $\Pi_b':=Evaluate_\cE(\lambda,T,\Pi_b)$.
\end{enumerate}
\end{definition}

A scheme is \emph{correct} if for any input state $\sigma$ with associated ciphertext $\Pi$, and any circuit $T$ chosen from a set of permitted quantum circuits $T_\cE$, $Decrypt_\cE(\lambda,\Pi_a,\Pi_b')=T(\sigma)$ (up to a possible global phase).  A scheme with perfect data privacy is one such that Bob cannot learn anything about the input $\sigma$, except the size of the input, from $\Pi_b$. The requirement of \emph{circuit privacy} refers to that Alice does not learn anything about the circuit performed beyond what can be deduced from the result of the computation itself.  In actual protocols we seem only able to approach this goal. A quantum fully homomorphic encryption (QFHE) scheme is one that is homomorphic for all circuits, i.e. the set $T_\cE$ is the set of all unitary quantum circuits. For approximate data privacy in QHE, there is a q-IND-CPA criterion \cite{bj15} which is related to the indistinguishability of encoded states. We regard interactive QHE schemes as those have similar security characteristics as the usual QHE schemes, but allow many rounds of communication. Naturally, any one party may choose to abort during the interactive protocol (but we exclude that in the scheme in Sec~\ref{sec6}). In Sec.~\ref{sec6}, we encounter a new criterion for data privacy in (interactive) QHE schemes in terms of mutual information for uniformly random inputs.

Here is a brief note about data privacy. It is well known [see for example Theorem 1(ii) of \cite{ygc10}] that if a unitary is exactly and deterministically implemented on all possible pure input states, no information about the input is leaked to the environment (which can be understood as the opposite party in our setting). Hence, data privacy is sort of trivial in some cases, but becomes an interesting problem when data privacy is more important than the correctness of computation, or when the computations are not restricted to unitaries, or when the type of input states is restricted (say real product inputs only). The QHE Schemes~\ref{schsub} and \ref{schmain} in this paper satisfy that the correctness of evaluation in general implies both data privacy and circuit privacy, and since arbitrary measurements can be done at the end, the computations here are not limited to unitaries. (The expression ``in general'' means excluding some classes of computations in which the output contains all information about the circuit, e.g. the special case in Scheme~\ref{schsub} where the polynomial has only one term.)  As far as we know, this phenomenon of ``correctness implies privacy'' is not previously known for non-unitary computations, and is thus a strong point of our schemes.

Denote the Pauli operators by $\Xgate=\sigma_x=\left(\begin{matrix} 0 & 1 \\ 1 & 0 \end{matrix}\right)$, $\sigma_y=\left(\begin{matrix} 0 & -i \\ i & 0 \end{matrix}\right)$, $\Zgate=\sigma_z=\left(\begin{matrix} 1 & 0 \\ 0 & -1 \end{matrix}\right)$, $\igate=\left(\begin{matrix} 1 & 0 \\ 0 & 1 \end{matrix}\right)$. The set of Clifford gates is defined as the set of unitaries $U$ satisfying $U P U^\dag=\pm P'$, where $P$ and $P'$ are tensor products of Pauli operators. The $n$-qubit Clifford gates form a group called Clifford group. Since we ignore the global phase here, the size of the $n$-qubit Clifford group is $\frac{1}{8}$ of an expression given in \cite{CRSS98}. In particular, the one-qubit Clifford group has $24$ elements. Let the rotations of angle $\theta$ about the three axes of the Bloch sphere be $\R_y(\theta)=\left(\begin{matrix} \cos\frac{\theta}{2} & -\sin\frac{\theta}{2} \\ \sin\frac{\theta}{2} & \cos\frac{\theta}{2} \end{matrix}\right)$, $\R_z(\theta)=\left(\begin{matrix} 1 & 0 \\ 0 & e^{i\theta} \end{matrix}\right)$, and $\R_x(\theta)=\left(\begin{matrix} \cos\frac{\theta}{2} & -i\sin\frac{\theta}{2} \\ -i\sin\frac{\theta}{2} & \cos\frac{\theta}{2} \end{matrix}\right)$. Let $\tgate=\left(\begin{matrix} 1 & 0 \\ 0 & e^{i\frac{\pi}{4}} \end{matrix}\right)$, $\hgate=\frac{1}{\sqrt{2}}\left(\begin{matrix} 1 & 1 \\ 1 & -1\end{matrix}\right)$, and $\pgate=\left(\begin{matrix} 1 & 0 \\ 0 & i\end{matrix}\right)$. Any one-qubit Clifford gate (ignoring global phase, same below) can be decomposed using products of $\R_y(\frac{\pi}{2})$ (proportional to $e^{i\frac{\pi}{4}\sigma_y}$) and $\R_z(\frac{\pi}{2})$ (proportional to $e^{i\frac{\pi}{4}\sigma_z}$) in different sequences. The $n$-qubit Clifford gates can be decomposed into products of one-qubit Clifford gates and the $\cnot$ (controlled-$\sigma_x$) gate. The Clifford gates and the $\tgate$ gate together form a universal set of gates for quantum computing. For more about universal sets of gates, see \cite{Shi03}. In this paper, the term ``EPR pair'' refers to two qubits in the state $\frac{1}{\sqrt{2}}\left(\ket{00}+\ket{11}\right)$. Let $\ket{+}$ denote $\frac{1}{\sqrt{2}}(\ket{0}+\ket{1})$, and $\ket{-}$ denote $\frac{1}{\sqrt{2}}(\ket{0}-\ket{1})$.

We shall adopt the rebit quantum computation formalism in \cite{Rudolph02} in our schemes. (See \cite{MMG09} for more about simulating the usual quantum circuit or Hamiltonian within the rebit formalism.) States of complex amplitudes can be represented by states of real amplitudes with only one extra qubit. Let there be an ancillary qubit with an orthonormal basis $\{\ket{R},\ket{I}\}$. The following non-unitary encoding is used in \cite{Rudolph02}:
\begin{eqnarray}\label{eq:rebit}
%&(a+b i)\ket{0}+(c+ d i)\ket{1} \longrightarrow\notag\\
%&a\ket{0}\ket{R}+b\ket{0}\ket{I}+c\ket{1}\ket{R}+d\ket{1}\ket{I},
(a+b i)\ket{0}+(c+ d i)\ket{1} \longrightarrow a\ket{0}\ket{R}+b\ket{0}\ket{I}+c\ket{1}\ket{R}+d\ket{1}\ket{I},
\end{eqnarray}
where $a,b,c,d$ are real numbers. When there are $n$ qubits in the initial input, only one ancillary qubit is needed. This ancillary qubit will be called the ``phase qubit'' in this paper. For Scheme~\ref{sch1} in this paper, we only need to deal with real input states, and in this case the initial encoding is trivial.

Next, we introduce how to apply gates on such representation. A unitary rotation $\R_z(\theta)$ in the $Z$ basis of an original qubit is mapped to a controlled-$\R_y(2\theta)$ on the real (i.e. original) qubit and the phase qubit, with the real qubit as the control. A unitary rotation in the $Y$ basis of an original qubit is still the same gate on the original qubit, and no action is needed on the phase qubit. An entangling two-qubit gate $F(\frac{\pi}{2})$ [which is controlled-$\R_y(\pi)$] is implemented by the same gate on the original qubits, and no action is needed on the phase qubit. These gates are all real (meaning that their matrix elements are real), and they can generate all possible unitary operations on any number of qubits, hence quantum computation in this particular representation can be done using only real gates. Let us denote three gates (on the logical qubit) $\R_z(\frac{\pi}{2})$,  $\R_y(\frac{\pi}{2})$ and $F(\frac{\pi}{2})$ as forming the gate set $S_0$. These three gates can generate the Clifford group on any number of logical qubits, and each of them can be implemented using real gates in the rebit formalism.

In the following paragraphs, we show how an unknown logical one-qubit Clifford gate can be implemented up to a phase by just performing the one-qubit gates in $S_0$ and the uncertain gates in the set $\{\igate,\R_y(\frac{\pi}{2}),\sigma_y,\R_y(3\frac{\pi}{2})\}$. We shall call the latter gate the uncertain $\R_y(\frac{k\pi}{2})$ gate.  The uncertain $\R_y(\frac{k\pi}{2})$ gate on Alice's side is generated using the gadget shown in Fig.~\ref{fig:Ygadget} which initially contains an EPR pair shared by Alice and Bob. Alice performs a controlled-$i\sigma_y$ gate with her part of the EPR pair as the control and the data qubit (assumed to be real) as the target. Then she performs a $\R_y(\frac{\pi}{2})$ gate on her qubit in the EPR pair and measures it in the $Z$ basis (with basis states being $\ket{0}$ or $\ket{1}$). Bob performs a $\R_y(\frac{j\pi}{2})$ gate and does a measurement in the $Z$ basis, where $j\in\{0,1\}$. The choice of $j$ determines the gate applied on Alice's qubit up to a correction $\R_y(\pi)=-i\sigma_y$ and a phase. Such correction can be commuted through the Clifford gates performed later with some Pauli operators as corrections.

Inspired by an $\hgate$-gate gadget construction in \cite{Broadbent18}, we give a construction using $\R_z(\frac{\pi}{2})$ and $\R_y(\frac{j\pi}{2})$ for some integers $j$ to implement one of the four gates $\R_z(\frac{k\pi}{2})$, where $k\in\{0,1,2,3\}$. We have
\begin{eqnarray}
\R_x(\frac{\pi}{2})&=&e^{-i\pi/4}\R_y(\frac{\pi}{2})\R_z(\frac{\pi}{2})\R_y(-\frac{\pi}{2}),\\
%\R_x(\frac{\pi}{2})&=&\R_z(-\frac{\pi}{2})\R_y(\frac{\pi}{2})\R_z(\frac{\pi}{2}),\\
%\R_z(-\frac{k\pi}{2})&=&e^{i\gamma_k}\R_x(-\frac{\pi}{2})\R_y(\frac{k\pi}{2})\R_x(\frac{\pi}{2}),\notag\\
%&&{\rm for}\quad k=0,1,2,3,
\R_z(-\frac{k\pi}{2})&=&e^{i\gamma_k}\R_x(-\frac{\pi}{2})\R_y(\frac{k\pi}{2})\R_x(\frac{\pi}{2}),\quad {\rm for}\quad k=0,1,2,3,
\end{eqnarray}
where $\gamma_k$ are real constants, and $\R_x(\frac{\pi}{2})=\frac{1}{\sqrt{2}}\left(\begin{matrix} 1 & -i \\ -i & 1 \end{matrix}\right)$. Noticing that $\R_x(-\frac{\pi}{2})=-\left[\R_x(\frac{\pi}{2})\right]^3$, the uncertain $\R_z(\frac{k\pi}{2})$ gate can be implemented using the uncertain $\R_y(\frac{k\pi}{2})$ gate and the gate $\R_z(\frac{\pi}{2})$ which is in $S_0$. Since $\R_z(\frac{\pi}{2})$ and $\R_y(\frac{\pi}{2})$ generate the single-qubit Clifford group, the two types of uncertain gates generate the same group.

We also notice that a $\tgate_y=\R_y(\frac{\pi}{4})$ gate can be implemented using the same gadget. Alice's action is still the same as before. Bob does some gate $\R_y(\frac{j\pi}{4})$ with $j=1$ or $3$ before doing a $Z$-basis measurement. The choice of $j$ depends on Alice and Bob's previous measurement results. All previous gadgets and fixed gates give rise to some Pauli operators known to Bob (which are to be corrected by Alice later), and these Pauli operators cause some Pauli operator $\Zgate$ on Bob's qubit in the gadget, which is corrected by Bob's choice of gate before measurement in the gadget. Generally, Bob may implement a $\R_y(\theta)$ gate for any real $\theta$, by suitably choosing his gate before measurement in the gadget. We omit the details here, since this is a special case of the more general multi-qubit $Y$-diagonal gate in Scheme~\ref{sch1} below.

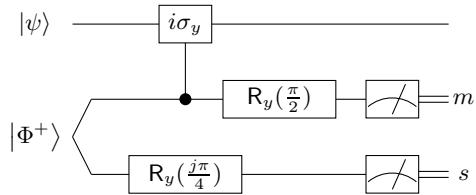
\begin{figure}[h]
\centering
\begin{tikzpicture}
\node at (-0.5,5) {$\ket{\psi}$};
\draw (0,5)--(5,5);
\node at (-0.5,3.5) {$\ket{\Phi^+}$};
\draw (0,3.5)--(.25,4);
\draw (0,3.5)--(.25,3);
\draw (0.25,4)--(4.25,4);
\draw (0.25,3)--(4.25,3);

\filldraw[fill=white] (1.15,5.25) rectangle (1.85,4.75);
\node at (1.5,5) {$i\sigma_y$};
\node at (1.5,4) {\cntrl};
\draw (1.5,4.75)--(1.5,4);

\filldraw[fill=white] (2.0,4.25) rectangle (3.5,3.75);
\node at (2.75,4) {$\R_y(\frac{\pi}{2})$};

\node at (4.25,4) {\meas};
\draw (4.6,4.04)--(5,4.04);
\draw (4.6,3.96)--(5,3.96);
\node at (5.2,4) {$m$};

\filldraw[fill=white] (0.75,3.25) rectangle (2.25,2.75);
\node at (1.5,3) {$\R_y(\frac{j\pi}{4})$};

\node at (4.25,3) {\meas};
\draw (4.6,3.04)--(5,3.04);
\draw (4.6,2.96)--(5,2.96);
\node at (5.2,3) {$s$};

\end{tikzpicture}
\caption{The gadget for implementing Y-axis rotations $\R_y(\frac{k\pi}{2})$ for integer $k$, or the $\tgate_y=\R_y(\frac{\pi}{4})$ gate on the first qubit up to a possible correction $\R_y(\pi)=-i\sigma_y$ and a possible phase depending on $j$ and the measurement results. The entangled state $\ket{\Phi^{+}}=\frac{1}{\sqrt{2}}\left(\ket{00}+\ket{11}\right)$. For implementing $\R_y(\frac{k\pi}{2})$ with integer $k$, the integer $j$ is $0$ or $2$, dependent only on the desired gate; for implementing the $\tgate_y$, the integer $j$ is $1$ or $3$, which depends on Alice's measurement outcome and the previous gates and measurement outcomes in the circuit. In Scheme~\ref{sch1}, the last qubit is on Bob's side and the other two qubits are on Alice's side.}\label{fig:Ygadget}
\end{figure}

\section{Two schemes based on rebits}\label{sec4}

{\bf (1) About Scheme~\ref{sch1}.}
The Scheme~\ref{sch1} is for a restricted type of circuits with restricted types of input states. The scheme uses the rebit quantum computation techniques in \cite{Rudolph02}. In the following, we denote the number of input data qubits as $n$.

The input state for Scheme~\ref{sch1} can be a product real state, but it can also be slightly more general: a product state of  one-qubit real states on $n-1$ qubits with a logical qubit state encoded in the rebit form \eqref{eq:rebit}, where the physical qubit for that exceptional logical qubit is the first data qubit. The output of Scheme~\ref{sch1} is of the form of \eqref{eq:rebit}. If the required output of Scheme~\ref{sch1} is classical, we only need to measure the qubits except the phase qubit to get the outcome. There might be other applications in which there are both real and imaginary parts in the output state, and the wanted part is real. Then we may still have some probability of success: do a measurement of the phase qubit in the $\{\ket{R},\ket{I}\}$ basis, and if the outcome corresponds to $\ket{R}$, we then measure the remaining qubits to get the result. If an output state with complex amplitudes is required, we may measure the phase qubit in the $\{(\ket{R}+\ket{I})/\sqrt{2},(\ket{R}-\ket{I})/\sqrt{2}\}$ basis, and if the result is the first one, we know the output state is correct.

The Scheme~\ref{sch1} allows continuous families of multi-qubit gates to be performed using relatively few gadgets. Bob is required to be able to perform more general operations than single-qubit measurements. The ability to perform two-qubit gates and single-qubit measurements is enough, but for better efficiency, he should do gates on more than two qubits.

Each layer of the circuit in the scheme implements some $Y$-diagonal unitary. Note that $Y$-diagonal unitaries on multiple qubits are not necessarily real, and the scheme is limited to implementing those real $Y$-diagonal unitaries. This is to guarantee the satisfaction of the requirement in rebit computation that the physical states are real. A class of examples of a three-qubit real $Y$-diagonal gate is $U=\cos(\theta)\igate\ox\igate\ox\igate+\sin(\theta)\R_y(\pi)\ox\R_y(\pi)\ox\R_y(\pi)$, where $\theta$ is real.

The implementation of a $Y$-diagonal unitary uses the group-based protocol of remote implementation of unitary operators in \cite{LW13}. Such protocol is derived from the protocol for implementation of double-group unitaries on two parties \cite{ygc10}. The double-group unitary protocol is a bipartite version of a group-based local quantum computing scheme \cite{KR03}, with additional analysis in the case of projective representations. In the current work, we only deal with the simple case of ordinary representation of a group, and the group is Abelian: it is the $K$-fold direct product of the $C_2$ group (the group with two elements), denoted as $C_2^K$. In Scheme~\ref{sch1}, $K$ is the number of qubits involved in the logical $Y$-diagonal unitary. The paper \cite{LW13} contains the detailed steps for the remote unitary protocol for the case of ordinary representation of a group. In Scheme~\ref{sch1}, all $K$ qubits are on Alice's side. The unitary operators used in the current case are $\bigotimes_{i=1}^K \sigma_{y,i}^{a_i}$, where $a_i\in\{0,1\}$, and the subscript $i$ means the gate acts on qubit $i$. The target unitary $U$ (on $K$ data qubits) is diagonal in the $Y$ axis, so it can be expanded using such a set of unitary operators:
\begin{eqnarray}\label{eq:Ucf}
U=\sum_f c(f) V(f),
\end{eqnarray}
where $f$ is a group element in $C_2^K$, and can be represented using a binary string of length $K$, consisting of the bits $a_i$: $f=\left(a_1,\dots,a_K\right)$. The unitary $V(f)$ is of the form $\bigotimes_{i=1}^K \sigma_{y,i}^{a_i}$. The coefficients $c_f$ appear in the unitary matrix $C$ in the remote unitary scheme. From the equation (45) in \cite{ygc10}, the matrix $C$ is defined using (ignoring phases since we are dealing with ordinary representations here)
\begin{eqnarray}\label{eq:ci}
C_{g,f}=c(g^{-1}f),
\end{eqnarray}
where $g^{-1}$ and $f$ are group elements, and their product is the group multiplication. When each element is represented using a binary string of length $K$ as mentioned above, the group multiplication is just the addition modulo 2 on each position of the vector.  The corrections are $V(g)^\dag$. In the current case, it amounts to $\igate$ or $\sigma_y$ [equivalent to $\R_y(\pi)$ up to a global phase] on each qubit. It is shown in \cite{KR03} that when $U$ is unitary and that the $\{V(f)\}$ is an ordinary representation, it is guaranteed that there exists at least one $C$ matrix that is unitary. Further, if the $\{V(f)\}$ is a linearly independent set (which is true for the current case), then $C$ is unique. These statements are extended in \cite{ygc10} to the case of projective representations.

Some note about a particular step in the scheme: Bob's $\pgate$ or $\pgate^\dag$ gates on his qubit in each gadget are the product of two gates: one of them is $\Zgate$ or $\igate$, which is a correction according to Alice's measurement outcome, and the other is a fixed $\pgate^\dag$ gate to correct for the phase $i$ in the controlled-$i\sigma_y$ gate in the gadget. Bob's such gates can also be absorbed into his unitary $C$.

The entanglement and classical communication costs in Scheme~\ref{sch1} scale linearly with the product of the input size and the circuit depth. For data privacy in Scheme~\ref{sch1}, we have the following theorem, with the proof in Appendix~\ref{appd1}.
\begin{theorem}\label{thm1}
When Scheme~\ref{sch1} is used for almost-commuting circuits with input state of the type specified in the scheme, the input data is perfectly secure.
\end{theorem}

A statement about the circuit privacy in Scheme~\ref{sch1} is in Theorem \ref{thm4} below. The circuit privacy is already optimal up to a constant factor, since Alice can always learn $O(n)$ bits of information about the circuit by looking at her output.

\begin{algorithm*}[htb]
\caption{QHE for almost-commuting circuits with product real input}\label{sch1}
\begin{flushleft}
{\bf The type of allowed circuits:} an \emph{almost-commuting circuit}. The logical circuit contains one or more real $Y$-diagonal unitaries which act on subsets of data qubits, interspersed with $\R_z(\frac{j\pi}{2})$  ($j=1,3$) on the first data qubit only. This implies that in the physical circuit, there may be multiple uses of the gate gadget on each data qubit, and the only other gates are the two-qubit controlled-$\R_y(j\pi)$ on only one of the data qubits and the phase qubit.\\
{\bf The type of allowed input states:} a product real state, or with the first data qubit replaced by a complex qubit encoded in the rebit form \eqref{eq:rebit}.\\
\end{flushleft}
\begin{enumerate}
\item Alice and Bob initially share some EPR pairs. Alice prepares the classical input state in the data qubits. Alice uses an ancillary qubit to be used as the phase qubit.  It is assumed that both parties know the general form and the fixed one-qubit gates in the desired \emph{logical} circuit, except that the details of the $Y$-diagonal unitaries are known only to Bob.

\item Alice implements the fixed gates and her part of the gate gadgets. The \emph{logical} $\R_z(\frac{j\pi}{2})$ ($j=1,3$) gates are implemented using the controlled-$\R_y(j\pi)$ gate on the real data qubit and her phase qubit. For a real $Y$-diagonal unitary acting on one or more qubits, it is implemented with the help of some gate gadgets (each on one data qubit) similar to the one in Fig.~\ref{fig:Ygadget} but with Bob's operations modified. Each gadget initially contains an EPR pair of qubits: $a$ on Alice's side, and $b$ on Bob's side. Alice performs a local controlled-$i\sigma_y$ gate with the qubit $a$ as the control, and the real data qubit as the target. She then does a $\R_y(\frac{\pi}{2})$ gate followed by a $Z$-basis measurement on the qubit $a$. She sends Bob all available measurement outcomes.

\item Bob receives Alice's message. He maintains a list of Pauli corrections for the data qubits and Alice's phase qubit, and initially they are all $\igate$. For each $Y$-diagonal unitary which may act on multiple qubits, he does the following: for each gadget, if Alice's measurement outcome in this gadget corresponds to $\ket{0}$, he does a $\pgate$ gate on his qubit in the gadget; otherwise he does a $\pgate^\dag$ gate on this qubit. When there are $\Zgate$ or $\Xgate$ in his list of Pauli corrections, he also does $\Zgate$ gates on the relevant qubits in the gadgets. He then performs a unitary transform $C$ given in Eq.~\eqref{eq:ci} which acts jointly on all his qubits in the relevant gate gadgets. He then does a measurement in the $Z$ basis on each of the transformed qubits. According to the measurement results, he updates his list of Pauli corrections [$\R_y(\pi)$ for each measurement outcome $1$]. He updates the list when passing each fixed gate or later $Y$-diagonal unitaries in the circuit. When he reaches the end of the circuit, he sends the required part of the list of Pauli corrections to Alice.

\item Alice applies Pauli corrections on the data qubits and the phase qubit (the correction on the phase qubit may be omitted depending on applications). The output state is in the form of \eqref{eq:rebit}.
\end{enumerate}
\end{algorithm*}

We note that there is a simplification of the scheme, using a method from \cite{Lai17}: for the last $n-1$ data qubits, Alice could choose to apply possible $\R_y(\pi)$ masks on them and send them to Bob before the protocol starts. In this way, the ancillary qubits for those data qubits can be saved. At the end of the protocol, Bob sends these qubits back, and Alice can recover this part of the output by undoing the $\R_y(\pi)$ masks. An alternative way of simplification is by combining some gate gadgets so that each of the last $n-1$ data qubits uses one gate gadget only.

{\bf (2) About Scheme~\ref{sch2}.} Based on Scheme~\ref{sch1}, we propose the following Scheme~\ref{sch2} which works for a large class of polynomial-depth circuits. Its data privacy is only partial. Its circuit privacy is similar to that in Scheme~\ref{sch1}.

\begin{algorithm*}[htb]
\caption{QHE for a large class of circuits with real input and partial data privacy}\label{sch2}
\begin{flushleft}
{\bf The type of allowed circuits:} the logical circuit contains one or more real $Y$-diagonal unitaries which act on all $n$ data qubits, interspersed with $\R_z(\frac{j\pi}{2})$  ($j=1,3$) on all of the data qubits. In the physical circuit, there are multiple uses of the gate gadget on each data qubit, interspersed with two-qubit controlled-$\R_y(j\pi)$ gates between a data qubit (at arbitrary position) and the phase qubit.\\
\noindent{\bf The type of allowed input states:} real input states, or complex input states encoded in the rebit form of  \eqref{eq:rebit}.\\
\end{flushleft}
\begin{enumerate}
\item Alice and Bob initially share some EPR pairs. Alice prepares the classical input state in the data qubits. Alice has an ancillary qubit to be used the phase qubit. The desired circuit is decomposed using real $Y$-diagonal unitaries and the \emph{logical} $\R_z(\frac{j\pi}{2})$ ($j=1,3$) gates. It is assumed that both parties know the general form and the fixed one-qubit gates in the desired circuit, except that the details of the $Y$-diagonal unitaries are known only to Bob.

\item Alice implements the fixed gates and her part of the gate gadgets. The \emph{logical} $\R_z(\frac{j\pi}{2})$ ($j=1,3$) gates are implemented using the controlled-$\R_y(j\pi)$ gate on the real data qubit and her phase qubit. For a real $Y$-diagonal unitary acting on all data qubits, it is implemented with the help of some gate gadgets (each on one data qubit) similar to the one in Fig.~\ref{fig:Ygadget} but with Bob's operations modified. Each gadget initially contains an EPR pair of qubits: $a$ on Alice's side, and $b$ on Bob's side. Alice performs a local controlled-$i\sigma_y$ gate with the qubit $a$ as the control, and the real data qubit as the target. She then does a $\R_y(\frac{\pi}{2})$ gate followed by a $Z$-basis measurement on the qubit $a$. She sends Bob all available measurement outcomes.

\item Bob receives Alice's message. He maintains a list of Pauli corrections for the data qubits and Alice's phase qubit, and initially they are all $\igate$. For each $Y$-diagonal unitary which may act on multiple qubits, he does the following: for each gadget, if Alice's measurement outcome in this gadget corresponds to $\ket{0}$, he does a $\pgate$ gate on his qubit in the gadget; otherwise he does a $\pgate^\dag$ gate on this qubit. When there are $\Zgate$ or $\Xgate$ in his list of Pauli corrections, he also does $\Zgate$ gates on the relevant qubits in the gadgets. He then performs a unitary transform $C$ given in Eq.~\eqref{eq:ci} which acts jointly on all his qubits in the relevant gate gadgets. He then does a measurement in the $Z$ basis on each of the transformed qubits. According to the measurement results, he updates his list of Pauli corrections [$\R_y(\pi)$ for each measurement outcome $1$]. He updates the list when passing each fixed gate or later $Y$-diagonal unitaries in the circuit. When he reaches the end of the circuit, he sends the required part of the list of Pauli corrections to Alice.

\item Alice applies Pauli corrections on the data qubits and the phase qubit (the correction on the phase qubit may be omitted depending on applications). The output state is in the form of \eqref{eq:rebit}.
\end{enumerate}
\end{algorithm*}

The entanglement and classical communication costs in the two schemes scale linearly with the product of the input size and the circuit depth. The circuit in Scheme~\ref{sch2} only consists of products of real $Y$-diagonal unitaries performed by Bob and some fixed type of unitaries performed by Alice. Numerical evidence suggests that whenever $n>1$  (where $n$ is the number of data qubits), the type of the allowed circuits at the logical level is not arbitrary. This is confirmed by Theorem~\ref{thm3} below. But the scheme at least allows a large class of unitary circuits. The data privacy in Scheme~\ref{sch2} is partial. From the proof of Theorem~\ref{thm1}, it can be found that the two input states of all $\ket{0}$ and all $\ket{1}$ on the data qubits (with the phase qubit in the same initial state $\ket{0}$) cannot be distinguished by Bob. More generally, we have the following theorem, with the proof in Appendix~\ref{appd1}.
\begin{theorem}\label{thm2}
(i) Suppose two input states in Scheme~\ref{sch2} satisfy that each is a product state across the bipartition between the $n$ data qubits and the phase qubit, and they are related to each other by a $\R_y(\pi)$ on each data qubit and a $\R_y(\theta)$ on the phase qubit, where $\theta$ is any real number, then Bob cannot distinguish these two input states at all.\\
(ii) In Scheme~\ref{sch2}, if $n$ is odd, the following logical qubit is perfectly secure: it is composed by encoding $\ket{0}$ and $\ket{1}$ into $\ket{0}^{\ox n}$ and $\ket{1}^{\ox n}$, respectively, on the data qubits, and the phase qubit is that in Scheme~\ref{sch2}. The first $n$ qubits of the two encoded states could also be any computational-basis state and $\Xgate^{\ox n}$ acting on such state.
\end{theorem}

It is interesting to note that Theorem~\ref{thm2}(ii) does not hold for the case of even $n$, even if we replace the ``logical qubit'' with ``logical rebit''. For example, when $n=2$, the two input states $\frac{1}{\sqrt{2}}(\ket{00}+\ket{11})\ket{0}$ and $\frac{1}{\sqrt{2}}(\ket{00}-\ket{11})\ket{0}$ (where the last $\ket{0}$ is the input state of the phase qubit) can be distinguished perfectly by Bob with a constant-depth circuit. This is because $\R_y(\pi)$ on an even number of qubits is not equal to the logical $\R_y(\pi)$ under the encoding in Theorem~\ref{thm2}(ii). On the other hand, the two input states $(\cos\alpha\ket{00}+\sin\alpha\ket{11})\ket{0}$ and $(\sin\alpha\ket{00}+\cos\alpha\ket{11})\ket{0}$ which are not always orthogonal to each other ($\alpha$ is real), cannot be distinguished by Bob at all, according to Theorem~\ref{thm2}(i). The considerations about data privacy leads to the following restrictions on the type of allowed circuits at the logic level.

\begin{theorem}\label{thm3}
In Scheme~\ref{sch2}, if $n>1$, the type of allowed circuits (at the logic level) does not include all circuits on $n$ qubits.
\end{theorem}
\begin{proof}
First, we prove for the case of even $n$. Consider the logical rebit defined by encoding $\ket{0}$ and $\ket{1}$ into $\ket{0}^{\ox n}$ and $\ket{1}^{\ox n}$, respectively. In the following we show that the logical Hadamard gate on such logical rebit cannot be implemented. Let the two input states in Scheme~\ref{sch2} be $\ket{0}^{\ox n}\ket{0}$ and $\ket{1}^{\ox n}\ket{0}$ (where the last qubit is the phase qubit), respectively, then after the logic Hadamard gate, the states on the data qubits become $\frac{1}{\sqrt{2}}(\ket{0}^{\ox n}+\ket{1}^{\ox n})$ and $\frac{1}{\sqrt{2}}(\ket{0}^{\ox n}-\ket{1}^{\ox n})$, respectively (up to unimportant global phases, which may appear as a change of state on the phase qubit). These two states are at least partially distinguishable by Bob with a constant depth circuit, according to explicit calculations in the case $n=2$ (note that these two input states are related to each other by a $\Zgate$ gate on only one qubit; the calculation is by finding an off-diagonal element in Bob's reduced density matrix that depends on the choice of input state), which is easily generalizable to any even $n$. This implies that the two original input states are at least partially distinguishable by Bob, which violates Theorem~\ref{thm2}(i). Therefore, the logic Hadamard gate on such logical rebit cannot be implemented deterministically.

Next, we prove for the case of odd $n$ where $n\ge 3$. The input states $\ket{0}^{\ox n}\ket{0}$ and  $\ket{0}^{\ox (n-1)}\ket{1}\ket{0}$ can be at least partially distinguished, according to explicit calculations.  (Note that these two input states are related to each other by an $\Xgate$ gate on the last data qubit.) This implies that the unitary that maps $\ket{1}^{\ox n}\ket{0}$ to  $\ket{0}^{\ox (n-1)}\ket{1}\ket{0}$ while keeping  $\ket{0}^{\ox n}\ket{0}$ unchanged cannot be implemented deterministically, since otherwise the two states $\ket{0}^{\ox n}\ket{0}$ and $\ket{1}^{\ox n}\ket{0}$ are at least partially distinguishable by Bob, which violates Theorem~\ref{thm2}(i). This completes the proof.
\end{proof}

For general complex input states encoded in rebit form, there is some evidence that partial data privacy exists in general. Specifically, there is numerical evidence that if classical information is encoded in some random (non-Pauli and non-real) product bases, and then encoded into the rebit form, Bob cannot generally distinguish the input states deterministically with a depth-$2$ circuit (depth-$2$ means there are two gadgets on each qubit line). Note that if the product bases above are chosen to be real, then Bob may deterministically distinguish between some pairs of input states with a depth-$2$ circuit, according to the argument in the proof of Theorem~\ref{thm1}.

For circuit privacy in the two schemes above, we have the following theorem. (The input size is $n$.)
\begin{theorem}\label{thm4}
If Scheme~\ref{sch1} (or Scheme~\ref{sch2}) is used for implementing a unitary, Alice learns at most $O(n)$ bits of information about the circuit; If it is used for computing a classical function with $k$ output bits, Alice learns at most $k$ bits of information about the circuit.
\end{theorem}
\begin{proof}
In \cite[Theorem 1]{DHL04}, consider the case that the initial maximum classical mutual information $I_c(\rho)=0$, which means two systems $A$ and $B$ satisfy that $\rho_{AB}=\rho_{A}\ox\rho_{B}$. After $l$ bits of classical communication in one direction, the maximum classical mutual information is at most $l$ bits. Consider adding a hypothetical program register in Scheme~\ref{sch1} or \ref{sch2} on Bob's side, and call it $C$. After Alice has sent measurement outcomes to Bob in the scheme, the combined system of Alice's system (denoted $A$) and $C$ is in a direct product (mixed) state. Then, after some local operations on Bob's side between $C$ and Bob's other systems, the systems $A$ and $C$ is still in a direct product (mixed) state. Then $l$ bits of classical message is sent by Bob to Alice. Applying \cite[Theorem 1]{DHL04} to this case, we get that the final maximum classical mutual information is at most $l$ bits. This means Alice learns at most $l$ bits of information about the program register $C$. Thus the assertions hold.
\end{proof}

In Scheme~\ref{sch2}, since Bob only sends $2n$ classical bits, Alice can learn at most $2n$ bits of information about the circuit. We think this amount is optimal for a sufficiently large class of circuits such as the one in Scheme~\ref{sch2}), since Alice could always use superdense coding to learn about $2n$ bits from the output of the computation.

\section{An interactive QHE scheme}\label{sec5}

Before describing the main interactive QHE scheme and its enhanced version, we first introduce the non-interactive Scheme~\ref{schsub}. It is almost a subroutine in the main Scheme~\ref{schmain}, except that it is altered a little when used in the main scheme.

{\bf (1) About Scheme~\ref{schsub}.}  We propose the following Scheme~\ref{schsub} for computing linear polynomials with classical input. It is for computing $y=(c+\sum_{i=1}^n a_i x_i)\,\,{\rm mod}\,\,2$, where $x_i$ are bit values of Alice's classical input, and $a_i$ and $c$ are constant bits known to Bob. The Scheme~\ref{schsub} uses $2nk$ EPR pairs in the teleportations, where $k$ is a positive integer that both parties agree on. The Alice's data $x_i$ in scheme are asymptotically secure for large $k$, and Bob's attempts to learn about the data would affect the correctness of evaluation. However, Bob's input (the $a_i$) are not secure: Alice could initially entangle her input qubits with some auxiliary qubits, and carry out the scheme as if she had used the usual input states, to learn some information about Bob's input $a_i$, at the cost of not computing the distributed outcome correctly (this refers to the situation when Scheme~\ref{schsub} is used in Scheme~\ref{schmain} below, in which the sending of the last bit by Bob is omitted, thus the outcome of the altered Scheme~\ref{schsub} is distributed across two parties), although she can compute the correct outcome herself. Such property is exploited below to yield a protocol with verifications such that Alice's cheating can be caught by Bob.

\begin{algorithm*}[htb]
\caption{Partially-secure QHE for computing classical linear polynomials}\label{schsub}
{\bf The type of allowed circuits:} those calculating a linear polynomial $y=(c+\sum_{i=1}^n a_i x_i)\,\,{\rm mod}\,\,2$, where $x_i$ are bit values of Alice's classical input, and $a_i$ and $c$ are constant bits known to Bob.\\
\begin{enumerate}
\item Alice and Bob agree on a positive integer $k$ related to the security of the scheme. A suitable choice of $k$ for a single use of the current scheme is $O(\log \frac{1}{\e})$, where $\e$ is a small real number indicating the targeted amount of leakage of each qubit of data.
\item
\begin{enumerate}
\item For each $i$, Alice randomly chooses $k$ bits $x_{ij}$ satisfying that $x_i=\sum_{j=1}^k x_{ij} \mod 2$, and she generates random bits $s_{ij}$, where $j=1,\dots,k$. The random bits $s_{ij}$ are unbiased and independent from each other and from the bits $x_{ij}$.
\item For each $(i,j)$ pair, Alice encodes $x_{ij}$'s value $0$ into $\ket{00}$ (when $s_{ij}=0$) or $\ket{++}$ (when $s_{ij}=1$), and $x_{ij}$'s value $1$ into $\ket{10}$ (when $s_{ij}=0$) or $\ket{+-}$ (when $s_{ij}=1$).
\item For each pair of qubits, she teleports them to Bob without telling him the bit for $Z$ correction on the first qubit nor the bit for $X$ correction on the second qubit. When $s_{ij}=0$, she records the bit for $X$ correction on the second qubit as $t_{ij}$, otherwise she records the bit for $Z$ correction on the first qubit as $t_{ij}$. Then $t_{ij}$ is a random bit.
\end{enumerate}

\item
\begin{enumerate}
\item For each $(i,j)$ pair, if $a_i=0$, Bob does a $\cnot$ gate on the two received qubits corresponding to $x_{ij}$, otherwise he does nothing.\label{stepcnot}
\item Bob teleports the resulting qubits to Alice, but withholding part of the information about the measurement outcomes: he calculates the XOR of the two measurement outcomes in the teleportation of each qubit, and calculate the XOR of the obtained two bits in each pair, and sends the resulting bit to Alice. Thus, Alice receives the two qubits with a possible $\sigma_y$ mask for each qubit, together with possible $\sigma_z$ masks on an even number of qubits in each pair of qubits. If the $\sigma_y$ correction is needed for one received qubit, we regard Bob's mask bit for that qubit as $1$, and otherwise we regard it as $0$.
\end{enumerate}
\item For each $(i,j)$ pair, Alice measures the two qubits in the $Z$ basis if $s_{ij}=0$, and in the $X$ basis if $s_{ij}=1$. She calculates the XOR of both outcomes, and take the XOR with $t_{ij}$, and obtains a bit. She flips this bit if $s_{ij}$ and Bob's sent bit for this pair are both $1$. Alice calculates the XOR of the obtained bits for all different $(i,j)$. Bob calculates the XOR value of $c$ and all his mask bits for the $2kn$ qubits, and sends the resulting bit to Alice. Alice corrects the output by a bit flip if the last sent bit from Bob is $1$, and the result is the final output.
\end{enumerate}
\end{algorithm*}

Note that it was previously thought that Bob should permute the two qubits in each pair with probability $\frac{1}{2}$ after step~\ref{stepcnot}, but such permutation is in fact redundant (in Alice's view) because of the withholding of partial information in the teleportations in the following step. Thus, such step has been omitted.

The data in Scheme~\ref{schsub} is partially secure for any fixed $k$. For any $(i,j)$ pair, it can be found that the average density operator (averaged over the possible values of $s_{ij}$ and $t_{ij}$) for the two-qubit state corresponding to each input bit $x_{ij}$ are not orthogonal for the cases $x_{ij}=0$ and $x_{ij}=1$. The trace distance of the two density operators is $\frac{1}{2}$. If Bob chooses a best POVM measurement (which can be chosen to be a projective measurement: to measure the first qubit in each pair of qubits in the $\Zgate$ basis, and the second qubit in the pair in the $\Xgate$ basis), he has probability $\frac{3}{4}$ of guessing the value of $x_{ij}$ correctly, and he has only probability $\frac{1}{2}$ of assuredly finding out the value of $x_{ij}$. And because of the effective independence among the $x_{ij}$ with different $j$ (given that $x_i$ is \emph{a priori} unknown to Bob), Bob's has probability $\frac{1}{2^k}$ of assuredly finding out $x_i$. Equivalently, the trace distance of the two density operators corresponding to $x_i=0$ or $x_i=1$ on the $2k$ qubits is $\frac{1}{2^k}$, which agrees with the result of numerical calculations. Thus, the probability of Bob correctly guessing $x_i$ is $1-(1-\frac{1}{2^k})/2=\frac{1}{2}+\frac{1}{2^k}$, which implies that he gets little information about $x_i$ when $k$ is large. Since measurements generally perturb the original state, it can be expected that the correctness of the computation cannot be guaranteed when these measurements are done, and in fact, in the current case, the requirement of exact correctness of the computation implies that Bob cannot get any information about the value of $x_{ij}$, and thus he cannot find out the $x_i$. If approximate correctness is required, Bob still cannot get much information about the value of $x_i$. Since the bits of the input, the $x_i$, are independently encrypted (as opposed to the correlated encoding in Sec.~\ref{sec6}), the data security automatically satisfies the IND-CPA criterion (with classical input but quantum measurements, same below), thus the main Scheme~\ref{schmain} below, which contains the modified Scheme~\ref{schsub} as a subprocedure, satisfies the q-IND-CPA criterion in \cite{bj15}. At the end of Sec.~\ref{sec6} below, we will mention a continuous class of variants of the current scheme, in which the data privacy is somewhat compromised (but still good for sufficiently random input), while the circuit privacy improves significantly.

For circuit privacy, we first discuss the case of Alice being honest, i.e. she exactly follows the protocol: in this case, Bob's sending of functions of measurement outcomes in the teleportations do not reveal useful information. Thus Alice learns at most $1$ bit of information about the circuit, for the same reason as shown in the proof of Theorem~\ref{thm4}.

In the following we discuss circuit privacy in the case of a cheating Alice. She could use some input state different from that specified in the scheme, or initially entangle her input qubits with some auxiliary qubits, and carry out the scheme as if she had used the usual input states, to learn some information about Bob's input $a_i$. Her best choice in the case of $n=1$ is to use the two qubit state $\frac{1}{2}(\ket{00}+\ket{01}+\ket{10}-\ket{11})$ for a particular $(i,j)$ pair. After Bob returns two qubits, she can find out Bob's input $a_i$ with certainty. But if Alice cheats, she generally cannot get a definite distributed outcome $a_i x_{ij}$ (which is expected to exist as the XOR of two bits at the two parties) related to the particular $x_{ij}$. Although she may learn $a_i$ by cheating and compute $a_i x_{ij}$ by herself, such computation result is on Alice's side only, while in the form of Scheme~\ref{schsub} used in Scheme~\ref{schmain} below, it is distributed as the XOR of two bits at opposite parties. If Bob asks Alice to send back her part of the (classical) result to combine with his corresponding mask bit (which is itself the XOR of two local bits), an error would sometimes occur. The technical reason that for those $j$ for which Alice cheated, she cannot definitely make the distributed outcome correct is because of the following: in the case of a conclusive measurement outcome of her POVM to distinguish Bob's two programs, she cannot deterministically and correctly find out Bob's values of the logical mask bit (the XOR of his two local mask bits) since that needs a measurement incompatible with that POVM she performed. (The previous statement is a somewhat nontrivial fact implied by that the two cases of zero or two $\sigma_z$ masks for the pair may appear with probability $\frac{1}{2}$ each.) Thus Alice would have significant probability of error in learning the logical mask bit of Bob's. Next, we discuss two different strategies of Alice's.

One strategy a cheating Alice may adopt is to aim for perfect knowledge about $a_i$. This can be achieved by using the type of input state mentioned in the last paragraph (or some other states), but then she would have significant probability of error in guessing the mask bit of Bob's. Since such mask bits are used in Scheme~\ref{schmain} below for the $\tgate$ gates, Alice's lack of knowledge about such bits would cause the evaluation to be incorrect, unless she chooses to deviate from the scheme once she has cheated by learning perfectly about some $a_i$: she calculates the intermediate result directly from the knowledge about $a_i$, and cheats later to learn about Bob's other coefficient(s), and calculates the new intermediate result, and continues to do this throughout the evaluation. Therefore, we may say that Alice's cheating may either affect the correctness of the evaluation, or cause deviation from the original protocol, and in both cases such cheating would be caught by Bob in an interactive scheme (Scheme~\ref{schcheck} below).

Another strategy for Alice is that she aims for partial knowledge about $a_i$ and partial knowledge of Bob's mask bit. She could use joint input states and joint measurements across different $j$ to increase her probability of success in learning both Bob's input and the overall mask bit of Bob's, which has been confirmed by numerical calculations. As a result of this type of attack, we estimate that Alice would have significant probability of guessing correctly both the output bit and Bob's overall logical mask bit for a particular variable $x_i$. Since a variable will appear for many times in the Scheme~\ref{schmain} below, with suitable choices of the $k$ (see below), the probability of Alice guessing correctly both the contribution to the output from the variable $x_i$, and the overall mask bit of Bob's related to that variable, can be upper bounded by a constant.

{\bf (2) About Scheme~\ref{schmain}.} The following Scheme~\ref{schmain} is an interactive scheme. It has asymptotic data privacy but the circuit privacy is not good. Later in Scheme~\ref{schcheck} we add some artificial verifications for improving circuit privacy.

In Scheme~\ref{schmain}, we denote an upper bound on the number of $\tgate$ gates in the desired circuit as $R$.
The scheme makes use of some ideas from the computationally-secure QHE scheme {\sf TP} in \cite{Dulek16}: the idea of teleporting the state to Bob's side to let him do the gates; and the garden-hose construction for doing a Clifford gate depending on the exclusive-OR of two bits at opposite parties (Appendix~\ref{appd2}). It also uses many instances of Scheme~\ref{schsub} (called the lower-level scheme here), except the very last small step in it, namely the correction of the output bit by Alice according to Bob's message. This is because in Scheme~\ref{schmain} we hope to reduce Bob's messages to Alice as much as possible, for improving circuit privacy. The general structure of the scheme in \cite{Dulek16} is that Alice first teleports the input state to Bob's side, but withholds the Pauli keys, and Bob performs the Clifford gates and the $\tgate$ gates in the desired circuit, while the corrections are calculated with the help of a classical HE scheme and some measurement on some EPR pairs (based on the ``garden hose'' technique developed in some related papers). In this work we get rid of the classical HE scheme, and only use the most elementary instance of the ``garden hose'' model. By combining Scheme~\ref{schsub} with the simple ``garden hose'' model in Appendix~\ref{appd2}, we obtain an interactive QHE scheme.

In the scheme, we make use of the following polynomials to update the initial Pauli keys. The $f_{a,i}$ and $f_{b,i}$ ($i=1,\dots,n$) are key-updating polynomials which are linear polynomials in $2n+4R$ variables. All variables are in the finite field $\mathbb{F}_2$. The first $2n$ variables are the initial Pauli keys $a,i$ and $b,i$ ($i=1,\dots,n$), corresponding to the Pauli operator $\bigotimes_{i=1}^n X_i^{a,i}Z_i^{b,i}$ applied on $n$ qubits. The other $4R$ variables are Alice's Bell-state measurement outcomes in the ``garden hose'' gadget in Appendix~\ref{appd2}. The output of the polynomials are $\{f_{a,i},f_{b,i}\}_{i=1}^n$, corresponding to the Pauli correction $\bigotimes_{i=1}^n X_i^{f_{a,i}}Z_i^{f_{b,i}}$. (If these $4R$ variables are not introduced, Alice would need to tell the Bell-state measurement outcomes to Bob, making the data less secure for a cheating Bob.)

The initial $2n$ variables are unchanged at each Clifford stage in the circuit. Bob changes the coefficients for them in the polynomials in instances of Scheme~\ref{schsub}. The way Bob changes the coefficients is because of some key-update rules, and we call them effective key-update rules below, since they do not change Alice's keys but rather change Bob's coefficients. The effective key-update rules for the first $2n$ variables under the action of Clifford gates are reversible (hence no new variables need to be introduced), and they can be easily obtained from the following relations:
\begin{eqnarray}\label{eq:keyupdate1}
&&\pgate\Xgate=i\Xgate\Zgate\pgate,\quad\quad \pgate\Zgate=\Zgate\pgate,\notag\\
&&\hgate\Xgate=\Zgate\hgate,\quad\quad\quad \hgate\Zgate=\Xgate\hgate,\notag\\
%&&\cnot_{12} (\Xgate_1^a \Zgate_1^b \ox \Xgate_2^c \Zgate_2^d)\notag\\
%&&= (\Xgate_1^a \Zgate_1^{b\oplus d} \ox \Xgate_2^{a\oplus c} \Zgate_2^d)\cnot_{12},
&&\cnot_{12} (\Xgate_1^a \Zgate_1^b \ox \Xgate_2^c \Zgate_2^d)= (\Xgate_1^a \Zgate_1^{b\oplus d} \ox \Xgate_2^{a\oplus c} \Zgate_2^d)\cnot_{12},
\end{eqnarray}
where the $\oplus$ is addition modulo 2, and in the gate $\cnot_{12}$, the qubit 1 is the control. The effective key-update rules for the first $2n$ variables under the $\tgate$ gate can be obtained from the relations
\begin{eqnarray}\label{eq:keyupdate2}
\tgate\Zgate=\Zgate\tgate,\quad\quad \tgate\Xgate=e^{-\pi i/4}\pgate\Xgate\Zgate\tgate.
\end{eqnarray}
More details about the key-update rules are in \cite{bj15,Dulek16}.

The Scheme~\ref{schmain} works for any quantum input, including states with complex amplitudes, and mixed states. There are $2n+R$ instances of Scheme~\ref{schsub} in Scheme~\ref{schmain}. There are $R$ instances of the gadget in Appendix~\ref{appd2} for correcting an unwanted $\pgate$ gate after a $\tgate$ gate, up to Pauli corrections.

\begin{algorithm*}[htb]
\caption{Partially-secure interactive QHE for general quantum input}\label{schmain}

\begin{enumerate}
\item Alice and Bob initially share some EPR pairs. Alice performs the local operations in teleporting the $n$ encrypted data qubits to Bob without sending any message to Bob. This means Bob instantly receives $n$ data qubits with unknown Pauli masks. She keeps a classical copy of the Pauli keys corresponding to the measurement results in teleportations. The two parties agree on a positive integer $k$ which is related to the security of the scheme.

\item For each of the $R$ Clifford stages before a $\tgate$ gate in the target circuit, the two parties do the following.

\begin{enumerate}
\item According to the desired circuit, Bob performs local Clifford gates on the ``received'' $n$ data qubits. For the $\tgate$ gate in the desired circuit, Bob does the $\tgate$ gate.
\item The two parties do an instance of Scheme~\ref{schsub} (with the sending of the last bit by Bob and the last bit flip by Alice in Scheme~\ref{schsub} omitted, and called ``the lower-level scheme'' below). This requires Bob to wait for Alice's message before he does quantum operations, but while waiting, he could do some local classical calculations for finding the coefficients in the polynomials (including the constant term $c$ in the polynomial) to be used in the instance of Scheme~\ref{schsub}. The coefficients are calculated based on Bob's previous measurement outcomes, his outcomes in previous instances of the lower-level scheme, and his knowledge about the desired circuit. Alice performs her first part of operations in this instance of the lower-level scheme (up to sending messages to Bob), with the $2n$ variables being the initial Pauli keys, and $4R$ variables being her Bell-state measurement outcomes (the outcomes for those measurements not yet carried out are regarded as zero). After receiving Alice's message, Bob performs his part of operations in the lower-level scheme, and this includes sending some messages to Alice which are functions of measurement outcomes in the teleportations. After receiving Bob's message, Alice does the remaining operations in the lower-level scheme. Each party obtains a bit as their outcome for the instance of the lower-level scheme.
\item The XOR value of the two bits obtained above indicates whether a $\pgate^\dag$ correction needs to be performed. The two parties perform the Bell-state measurements in the garden-hose gadget in Appendix~\ref{appd2} locally, without any communication. (Since Alice's Bell-state measurements, not including the $\pgate^\dag$ gates, are independent of her input bit, there are only $4$ bits for the outcomes of Alice's Bell-state measurements in each use of the gadget.)
\end{enumerate}

\item Bob does the final stage of Clifford gates in the desired circuit on the encoded data qubits. He does local operations in teleporting the $n$ encoded data qubits to Alice, but without sending her the measurement outcomes in the teleportation.\label{step:beforelastprocess}

\item Alice and Bob perform the following procedure $2n$ times in parallel, each for calculating a Pauli key on one of the $n$ output qubits in the higher-level scheme. The procedure is as follows: Bob does classical calculations to find a polynomial. The coefficients of the polynomial are determined from Bob's previous measurement outcomes, his outcomes in previous instances of the lower-level scheme from step 2 above, and his knowledge about the desired circuit. Alice performs the first part of operations in an instance of the lower-level scheme (up to sending messages to Bob) with the input data being her $2n$ initial encryption keys and $4R$ bits which are outcomes of her previous Bell-state measurements. After receiving Alice's message, Bob performs the gates and measurements in the instance of the lower-level scheme. Bob corrects his part of the correction key (one bit) using some measurement outcome in teleportation in the previous step~\ref{step:beforelastprocess}, and sends the resulting bit to Alice. Alice completes the last part of the lower-level scheme and uses the bit from Bob's message to correct the outcome, and regards the result as a Pauli correction key in the higher-level scheme.

\item Alice applies Pauli gates on the data qubits according to the keys from the last step, to get the output state.
\end{enumerate}
\end{algorithm*}

The correctness of Scheme~\ref{schmain} follows from the correctness of Scheme~\ref{schsub} and our way of using the scheme {\sf TP} in \cite{Dulek16}. Since each input variable appears for $O(n+R)$ times (for example, in each of the $R$ polynomials related to $\tgate$ gates), it is necessary to set the parameter $k$ to be sufficiently large, while the circuit size needs to be multiplied by a factor to allow recompilation of the original circuit into a form not easily guessed by Alice from partial information about the coefficients in the polynomials. In the following, we give some estimate under the requirement of asymptotic data privacy; the circuit privacy is to be improved in Scheme~\ref{schcheck}. As mentioned previously, it is necessary to set $k=O(\log \frac{1}{\e})$ to hide the data in one instance of Scheme~\ref{schsub}. Since each input variable of Alice's appear for at most $O(n+R)$ times in Scheme~\ref{schmain}, to let the data privacy be independent of the correctness of evaluation, it suffices to set $k=O[\log(\frac{n+R}{\e})]$, where $\e$ is a small positive real number describing the leak of information about an input qubit of Alice's. An $\e$ near zero means that the qubit is almost secure. The total classical communication cost and entanglement cost are $O[(n^2+n R)\log(\frac{n+R}{\e})]$ (bits and ebits, respectively), as there are $O(n+R)$ instances of the lower-level scheme, each of which uses $O(nk)=O[n\log(\frac{n+R}{\e})]$ EPR pairs and a proportional amount of classical communication.

We now discuss the circuit privacy in Scheme~\ref{schmain}. If Alice honestly follows the protocol, the bits sent by Bob during the instances of the lower-level scheme do not actually carry useful information about his circuit, thus she learns at most $O(n)$ bits about the circuit, according to the analysis after Scheme~\ref{schsub}. This degree of circuit privacy is already optimal for general unitary circuits, since Alice could always use dense coding to learn about $2n$ bits from the output of the computation; and since we do not allow the use of initial entangled states for an honest Alice, she in fact learns only $n$ bits. If Alice cheats by using some other input state for each instance of the lower-level scheme, she could learn more information about the circuit, see the discussion after Scheme~\ref{schsub}. A method for improving circuit privacy is in the following Scheme~\ref{schcheck}.

\begin{algorithm*}[htb]
\caption{Interactive QHE with embedded verifications for improving circuit privacy}\label{schcheck}

\begin{enumerate}
\item It is assumed that Alice and Bob agree on an instance of Scheme~\ref{schmain}, that is, each party knows what he or she should know about the instance of the scheme.
\item Bob revises the original circuit by adding a constant number of auxiliary qubits initially in Pauli eigenstates on his side, and let these qubits undergo some joint gates with the data qubits, where the gates are chosen randomly and include some $\tgate$ gates, satisfying that each added qubit becomes a Pauli eigenstate independent of Alice's input data at some stage in the circuit. The circuit in the added part of the protocol is implemented in the same way as in Scheme~\ref{schmain}. In the redesigned protocol, Alice sends Bob some number of bits, each of which is her part of the distributed outcome of an added instance of reduced Scheme~\ref{schsub} (``reduced'' means without Bob sending her the last bit and without her subsequent correction), and they represent Pauli masks (due to the combined effect of some Bell-state measurements in some instances of the garden-hose gadget in Appendix~\ref{appd2}) for some auxiliary qubits at various steps in the protocol, for him to check if Alice has cheated.
\item Alice and Bob perform the designed modified protocol. At intermediate steps, Bob measures some local auxiliary qubits, and combines the outcomes with the Pauli masks sent by Alice, and then checks the result against his mask bits (which are from Scheme~\ref{schsub}) using the XOR operation, to check if their correlations (XOR of the corresponding bits) are as Bob expected. If not, then he regards Alice as having cheated, and aborts the protocol.
\item If the protocol is completed without Bob aborting, then the final result on Alice's side is the desired computational result.
\end{enumerate}
\end{algorithm*}

{\bf (3) About Scheme~\ref{schcheck}.} The Scheme~\ref{schcheck} is based on Scheme~\ref{schmain}. The general idea is that Bob asks Alice to return some messages as outputs for his embedded checking procedures, and Bob aborts the protocol if the returned messages are not as expected. In this way, Bob verifies the quantum computation for the original circuit. This is inspired by the results in the literature that a quantum verifier can verify the result of a polynomial-sized quantum computation of a prover. Some protocols of verification for different types of blind quantum computing have been constructed \cite{FK17,ABOEM17,Broadbent18}. The verification protocol \cite{FHM18} is not intended for blind quantum computing. A computationally-secure QFHE scheme with verification is given in \cite{ADSS17}.

The reason that verifications are of significant help here is because of a property of Scheme~\ref{schmain}: Alice's knowledge about (the current part of) Bob's program does not mean she can send back the correct messages. This is because Bob verifies Alice's messages against his mask bits, which are unknown to the honest Alice, and a cheating Alice cannot (almost) deterministically find out those mask bits while (almost) deterministically learning about the related part of Bob's program at the same time, since the two tasks involve incompatible measurements. An error of Alice, which can also be stated as an error of the distributed result in an instance of Scheme~\ref{schsub}, will affect the correctness of removing unwanted $\pgate$ after each $\tgate$ gate. Consequently, due to the $\tgate$ gates added in the random gates in Scheme~\ref{schcheck}, Alice has large probability of returning some wrong classical message, resulting in wrong state of Bob's auxiliary qubits after he corrects such qubits according to Alice's message.

Here is an explanation for the ``added instance of reduced Scheme~\ref{schsub}'' in Scheme~\ref{schcheck}. Such instance is for computing a polynomial whose output is a bit, representing whether some Pauli $Z$ or $X$ correction should be performed on some qubit on Bob's side, and Bob will later measure those qubits in some suitable Pauli basis to find out whether Alice has cheated. The input variables for such polynomial should be the (at most) $2n+4R$ variables in Scheme~\ref{schmain}, but since the verification is supposed to be independent of the input data, the $2n$ variables actually have zero coefficients (for honest Bob), and the remaining (at most) $4R$ variables are from Alice's Bell-state measurements in some instances of the garden-hose gadget in Appendix~\ref{appd2}.

By choosing the level of data privacy in the original scheme suitably, the Scheme~\ref{schcheck} keeps the original data qubits quite secure with the exception that the verification messages may reveal information about the data. Since the total length of such messages is limited to a constant (see the expression of $m$ below), the amount of information leak introduced by the revision compared to Scheme~\ref{schmain} is a constant. Thus, the data privacy may be quite satisfactory, especially in the applications in which partial information about the data do not reveal much useful information, say when the input is already the output of some hashing procedure (or the encryption procedure in some classical homomorphic encryption scheme) acting on some original data. By choosing a suitable number of check bits from Alice, the circuit privacy in Scheme~\ref{schcheck} can be enhanced to a satisfactory level: suppose there are $m$ check bits, then since a cheating Alice would return a wrong bit with some constant probability (a cheating Alice should cheat at almost all places in the protocol, otherwise she learns little about the circuit due to the possible different decompositions of the circuit; see also the analysis about circuit privacy after Scheme~\ref{schsub}), the probability that at least one check bit is wrong is about $1-c^m$, where $c\in(0,1)$ is a constant real number. If the desired level of the probability that a cheating Alice is not caught is $\delta$, we may set $m=O(\log \frac{1}{\delta})$. Consequently, the Scheme~\ref{schcheck} keeps all but $O(\log \frac{1}{\delta})$ qubits of the input secure, and there is $1-c'\delta$ chance that the information about Bob's intended circuit is asymptotically secure except for an amount proportional to input size, where $c'$ is a constant related to the prior probability that Alice may cheat.

Random recompilation of the circuit may help circuit privacy. For some simple types of circuits, such as Clifford circuits, it may be helpful to recompile the circuit to fit it into a more complex class of circuits. Enlarging the circuit size by a constant factor is already sufficient to allow for many different forms of the original circuit.

Note that a purely classical client may be used in the verification of quantum computation \cite{Mahadev18v,ADSS17} (at least in the case of classical input and output), but this is based on computational assumptions, thus it is out of the situations considered in Scheme~\ref{schcheck}. We leave as an open problem whether the method in \cite{Mahadev18v} or \cite{ADSS17} together with Scheme~\ref{schmain} can be used to construct an interactive QHE scheme with classical client with better security than that in \cite{Mahadev17}.

\section{A class of variants of Scheme~\ref{schsub}}\label{sec6}

In this section, we first introduce Scheme~\ref{schsub2}, which is a variant of Scheme~\ref{schsub}. The use of such scheme is discussed in later sections.
We then mention that there is a continuous class of variants of Scheme~\ref{schsub}, which interpolates between Scheme~\ref{schsub} and Scheme~\ref{schsub2}.

\begin{algorithm*}[htb]
\caption{A variant of Scheme~\ref{schsub} for computing classical linear polynomials}\label{schsub2}
{\bf The type of allowed circuits:} those calculating a linear polynomial $y=(c+\sum_{i=1}^n a_i x_i)\,\,{\rm mod}\,\,2$, where $x_i$ are bit values of Alice's classical input, and $a_i$ and $c$ are constant bits known to Bob.\\
\begin{enumerate}
\item Alice and Bob agree on a positive integer $k$ related to the security of the scheme. The suitable choice of $k$ is $O(n)$ in general, but may be chosen to be $O(\log n)$ if correlations among input bits may be leaked (c.f.~Theorem~\ref{thm5}).
\item
\begin{enumerate}
\item For each $i$, Alice randomly chooses $k$ bits $x_{ij}$ satisfying that $x_i=\sum_{j=1}^k x_{ij} \mod 2$. She generates independent and unbiased random bits $s_j$, where $j=1,\dots,k$. The $s_j$ are independent from $x_{ij}$.
\item For each $(i,j)$ pair, Alice encodes $x_{ij}$'s value $0$ into $\ket{00}$ (when $s_j=0$) or $\ket{++}$ (when $s_j=1$), and $x_{ij}$'s value $1$ into $\ket{10}$ (when $s_j=0$) or $\ket{+-}$ (when $s_j=1$).
\item For each pair of qubits, she teleports them to Bob without telling him the bit for $Z$ correction on the first qubit nor the bit for $X$ correction on the second qubit. When $s_j=0$, she records the bit for $X$ correction on the second qubit as $t_{ij}$, otherwise she records the bit for $Z$ correction on the first qubit as $t_{ij}$. Then $t_{ij}$ is a random bit.
\end{enumerate}

\item
\begin{enumerate}
\item For each $(i,j)$ pair, if $a_i=0$, Bob does a $\cnot$ gate on the two received qubits corresponding to $x_{ij}$, otherwise he does nothing.
\item Bob teleports the resulting qubits to Alice, but withholding part of the information about the measurement outcomes: he calculates the XOR of the two measurement outcomes in the teleportation of each qubit, and calculate the XOR of the obtained $2n$ bits with the same $j$, and sends the resulting bit, denoted $v_j$, to Alice. Thus, Alice receives the qubits with a possible $\sigma_y$ mask for each qubit, together with possible $\sigma_z$ masks on an even number of qubits in each group of $2n$ qubits. If the $\sigma_y$ correction is needed for one received qubit, we regard Bob's mask bit for that qubit as $1$, and otherwise we regard it as $0$.
\end{enumerate}
\item For each $j$, Alice measures the $2n$ qubits in the $Z$ basis if $s_j=0$, and in the $X$ basis if $s_j=1$. She calculates the XOR of the $2n$ outcomes for the same $j$, and take the XOR with $(\sum_{i=1}^n t_{ij}) \mod 2$, and obtains a single bit for each $j$. She flips the resulting bit if $s_j=1$ and Bob's sent bit $v_j=1$. Alice calculates the XOR of the obtained bits for all different $j$. Bob calculates the XOR value of $c$ and all his mask bits for the $2kn$ qubits, and sends the resulting bit to Alice. Alice corrects the output by a bit flip if the last sent bit from Bob is $1$, and the result is the final output.
\end{enumerate}
\end{algorithm*}

Since Scheme~\ref{schsub2} is non-interactive, from the proof of Theorem~\ref{thm4}, we find that a \emph{cheating} Alice learns at most $k + 1$ bits of information about Bob's circuit (i.e. about $a_j$ and the constant $c$). If the last bit of communication is omitted, which may be the case when Scheme~\ref{schsub2} is used in some higher-level schemes such as Scheme~\ref{schmain}, Alice learns at most $k$ bits of information about Bob's coefficients $a_j$, and in such case the result is distributed as the XOR of two remote bits. It should be stressed that Alice's learning of the circuit would affect the correctness of evaluation of the distributed result, for similar reasons as discussed after Scheme~\ref{schsub}. If the result is not distributed, but completely on Alice's side, and if $k+1<n$, she still cannot learn all information about Bob's input. If Alice wants to ensure that the evaluation (with distributed result when $k+1\ge n$) is exactly correct, she cannot learn any information about the circuit apart from what is learnable from the output of the evaluation. But she can learn about Bob's input up to the amount mentioned above if she does not care about the correctness of evaluation of the distributed result.

For the data privacy of Scheme~\ref{schsub2}, we consider two criteria: first, when the input is \emph{a priori} uniformly random in the other party's view, the amount of information learnable by Bob is small. The second is the usual IND-CPA type of criterion: that any two classical input strings are hardly distinguishable by Bob. The security under the first criterion is partially discussed in Theorem~\ref{thm5} and in Appendix~\ref{appd3}. It is found that even with uniformly random data, the security of Scheme~\ref{schsub2} is not very good by itself, but with the help of the type of verifications in the last section, the security in an interactive QHE scheme would be quite interesting. Later we show that the data security of Scheme~\ref{schsub2} does not satisfy the second criterion.

\begin{theorem}\label{thm5}
The following statements hold for Scheme~\ref{schsub2} in the case that Alice's input data is uniformly distributed in Bob's view.\\
(i) The information learnable by a cheating Bob about each bit of Alice's input is $\frac{1}{2^k}$ bits (although he may learn more about the correlations between the input bits).\\
(ii) The information learnable by a cheating Bob about the joint distribution of the input bits is at least $n-k$ bits.
\end{theorem}
\begin{proof}
(i) It is shown in the discussion after Scheme~\ref{schsub} that the trace distance of Bob's two reduced density matrices (on $2k$ qubits corresponding to $x_i$) for the different values of a particular $x_i$ is exactly $\frac{1}{2^k}$. Bob actually has a very good measurement to find out about $x_i$ in Scheme~\ref{schsub2} (also for Scheme~\ref{schsub}). He just needs to measure the first qubit in each pair of qubits in the $\Zgate$ basis, and the second qubit in the pair in the $\Xgate$ basis. Such measurement basis in the combined Hilbert space of $2k$ qubits is a fixed basis, in which the density matrix is diagonal. The mutual information between Alice and Bob (for uniformly random input $x_i$) by measuring the $2k$ qubits in this way is exactly $\frac{1}{2^k}$ bits: it is because the only cases in which the $x_i$ can be determined are that the qubit pairs are all in the state $\ket{0+}$ or all in the state $\ket{1-}$, and in all other cases (which are perfectly distinguishable from the two cases just mentioned) the two possible bit values of $x_i$ are equally likely. An alternative way to calculate the amount of mutual information is by using the fact that the Holevo quantity is equal to the classical mutual information in the current case, see Appendix~\ref{appd3}.

Since the input distribution is uniformly random, each $x_i$ is $1$ with probability exactly $\frac{1}{2}$, then, if Bob had no prior knowledge about the $\{s_j\}$, the $2kn$ qubits received by Bob are in a maximally mixed state in his view, for each possible combination of the values of $s_j$. In other words, Bob's information about the $\{s_j\}$ cannot be increased at all if he had zero information to start with. And he indeed initially has zero information about the $\{s_j\}$. Hence, Bob is completely ignorant of the $\{s_j\}$ even after his measurements. Therefore, his best strategy to learn about one of the bits $x_i$ is to learn it by measuring the corresponding $2k$ qubits. As mentioned above, the information about the bit $x_i$ learnable by Bob is $\frac{1}{2^k}$ bits for uniformly random input. By doing the same measurement for all input bits, and looking at the correlations among measurement outcomes across different $i$, he may learn about the joint distribution of the input bits.

(ii) If Alice further sends $k$ bits of information about the $\{s_j\}$ to Bob, he would know all information about the input $\{x_i\}$ by the following fixed-basis measurement below. He just needs to measure the first qubit in each pair of qubits in the $\Zgate$ basis, and the second qubit in the pair in the $\Xgate$ basis. The fixed-basis measurement and the construction of the scheme imply that the state received by Bob is effectively classical, and since locking does not happen for classical states \cite{DHL04}, we get that Bob's information about the $\{x_i\}$ before being told anything about the $\{s_j\}$ is at least $n-k$ bits.
\end{proof}

For more details about the classical mutual information when Alice's input is uniformly random, see Appendix~\ref{appd3}.

As a step towards studying the security under the second criterion, we prove the following result.
\begin{theorem}\label{thm6}
In Scheme~\ref{schsub2}, the trace distance between Bob's reduced density matrices for two different input binary strings of length $n$ is a constant dependent only on $n$ and $k$.
\end{theorem}
\begin{proof}
Without loss of generality, let the first input binary string be the all-zero string. Now assume the second input binary string contains only one $1$, then without loss of generality it can be assumed to be $0\dots 001$. Let the trace distance between Bob's reduced density matrices for these two input strings be $c_0$, which depends on $n$ and $k$. Now we consider the case that the second input string contains two $1$'s, and without loss of generality the string can be assumed to be $0\dots011$. The trace distance between Bob's reduced density matrices for the all-zero string and the above string can be expressed as $\frac{1}{2}\tr\sqrt{(\rho-\sigma)^2}$. For any unitary $U$ acting on such reduced density matrices, the expression above is equal to $\frac{1}{2}\tr\sqrt{(U\rho U^\dag-U\sigma U^\dag)^2}$. Now let $U$ be the following operation: Bob does $k$ $\cnot$ gates with the controlling qubit being the first qubit in each of the $k$ pairs of qubits corresponding to the input bit $x_n$, and the target qubit being the first qubit in each of the $k$ pairs of qubits corresponding to the input bit $x_{n-1}$; Bob also does $k$ $\cnot$ gates with the controlling qubit being the second qubit in each of the $k$ pairs of qubits corresponding to the input bit $x_{n-1}$, and the target qubit being the second qubit in each of the $k$ pairs of qubits corresponding to the input bit $x_n$. The $U$ effectively changes the input $0\dots001$ to $0\dots011$, while leaving the all-zero input string unchanged. Thus, the trace distance between Bob's reduced density matrices for the all-zero string and the $0\dots011$ is equal to $c_0$. By doing similar unitaries for other pairs of states corresponding to the all-zero input string and some nonzero input string, we get that the trace distance between Bob's reduced density matrices for the all-zero string and any nonzero string is equal to $c_0$. Then, noting that $R_y(\pi)$ on a pair of qubits corresponding to one of the input bits $x_i$ effectively flip the input bit $x_i$, we have that the trace distance between Bob's reduced density matrices corresponding to any pair of different input binary strings is $c_0$. This completes the proof.
\end{proof}

Numerical calculations suggest that the trace distance between Bob's reduced density matrices for the all-zero input string and for a mixture of other input strings is also equal to the same constant in Theorem~\ref{thm6}.

Now we discuss whether Scheme~\ref{schsub2} is secure under the IND-CPA criterion. While it is expected that the all-zero and the all-one inputs are more distinguishable compared to the case of only one bit input ($0$ or $1$), which has been confirmed by numerical calculations, it may seem counter-intuitive that many identical bits in the case of the zero string versus $0\dots001$ still help distinguishability compared to the one-bit input case. But there is an easy explanation for this phenomenon: since Bob knows the value of the input bit $x_i$ for a particular $i$, he may do some measurements in the $2k$ qubits corresponding to $x_i$ so that he learns something about Alice's choices of the encoding bases, i.e. the $s_j$, and such partial information helps the distinguishability when he measures other qubits. A very good measurement strategy for Bob is to measure the first qubit in each pair of qubits in the $\Zgate$ basis, and the second qubit in the pair in the $\Xgate$ basis. In this way, since he knows the value of $x_i$ for any $i>1$, he may exclude about a half of the possible strings $(s_j)$ by measuring the qubits for a particular $i$. After he repeats this procedure for many different $i$, he gets a lot of information about $(s_j)$ such that he can measure the qubits corresponding to $x_1$, and get quite a good amount of information about $x_1$, which determines the input is the zero string or the $0\dots001$. Since his ignorance about $(s_j)$ decreases exponentially with $n$, we may expect that he can distinguish the zero string or the $0\dots001$ with probability near $1$, when $n$ is large. And according to Theorem~\ref{thm6}, this means the trace distance between any two different inputs is nearly $1$ for large $n$, implying that the Scheme~\ref{schsub2} is not secure under the IND-CPA criterion. As a result, the overall interactive QHE scheme based on Scheme~\ref{schsub2} is not secure under the q-IND-CPA criterion in \cite{bj15}.

By using Scheme~\ref{schmain} as the main scheme, we obtain an interactive QHE scheme. The security depends on the parameter range: $k<n$ is guarantees circuit privacy, by using many $\tgate$ gates in the circuit and allowing recompilation of the circuit, but the data privacy is not good; $k\gg n$ is good for data privacy in the case of uniformly random data, but the circuit privacy is not good by itself, but similar to the last section, we may still in principle use Scheme~\ref{schcheck} for improving circuit privacy by embedded verifications, at some minor cost of affecting data privacy. Some further discussions about the uses of Scheme~\ref{schsub2} in bipartite secure computation and interactive QHE are in Sec.~\ref{sec8}.

A last comment is that the schemes in Sec.~\ref{sec5} may be improved by partially using the method in this section: in Scheme~\ref{schsub}, all the $s_{ij}$ are independently random. If we change the choices such that about $m$ bits $s_{ij}$ are the same, for the same $j$ but different $i$, then we may reduce the amount of classical communication from Bob to Alice by a factor of about $m$. This should help circuit privacy when Scheme~\ref{schsub} is used as a subprocedure in Scheme~\ref{schmain} and later in Scheme~\ref{schcheck}, while the data privacy is somewhat compromised, but if $m$ is small enough, say $m=o(\log_2 k)$, then the data security is still of the IND-CPA type, according to the remarks after Theorem~\ref{thm6} above.

\section{A scheme for computing classical linear polynomials based on information locking}\label{sec7}

In this section, we introduce Scheme~\ref{schsub3}, which is a variant of Scheme~\ref{schsub2} and Scheme~\ref{schsub}. It makes use of the method of information locking by quantum techniques \cite{DHL04}. An application of Scheme~\ref{schsub3} is discussed in the next section.

\begin{algorithm*}[htb]
\caption{A scheme for computing classical linear polynomials using quantum data locking}\label{schsub3}
{\bf The type of allowed circuits:} those calculating a linear polynomial $y=(c+\sum_{i=1}^n a_i x_i)\,\,{\rm mod}\,\,2$, where $x_i$ are bit values of Alice's classical input, and $a_i$ and $c$ are constant bits known to Bob.\\
\begin{enumerate}
\item Alice and Bob agree on a positive integer $k$ related to the security of the scheme. The suitable choice of $k$ is $O(n)$ in general, but may be chosen to be $O(\log n)$ if correlations among input bits may be leaked.
\item
\begin{enumerate}
\item For each $i$, Alice randomly chooses $k$ bits $x_{ij}$ satisfying that $x_i=\sum_{j=1}^k x_{ij} \mod 2$. She generates independent unbiased random bits $s_j$ and $t_j$, where $j=1,\dots,k$. The $s_j$, $t_j$ are independent from $x_{ij}$.
\item For each $(i,j)$ pair, Alice encodes $x_{ij}$'s value $0$ into $\ket{0}$ (when $s_j=0$) or $\ket{+}$ (when $s_j=1$), and $x_{ij}$'s value $1$ into $\ket{1}$ (when $s_j=0$) or $\ket{-}$ (when $s_j=1$).
\item For each $j\in\{1,\dots,k\}$, Alice encodes $t_j$ into a qubit using the rule in the previous step.
\item Alice teleports the $k(n+1)$ qubits to Bob.
\end{enumerate}

\item
\begin{enumerate}
\item Bob denotes $w:=\sum_{i=1}^n a_i \mod 2$.
\item For each $j\in\{1,\dots,k\}$, Bob does the following: if all $a_i$ are zero, he sets $u_j:=0$ and $v_j:=0$. Otherwise, he arranges the received qubits corresponding to $x_{ij}$ satisfying $a_i=1$ into pairs (with one unpaired qubit when $w=1$), and does a $\cnot$ gate on each pair, with the qubit of smaller index $i$ as the control; if $w=1$, Bob does a $\cnot$ gate on the unpaired qubit with the qubit corresponding to $t_j$, with the former as the control. Bob measures the first qubit in each pair from in the $\Xgate$ basis, and records the XOR (sum modulo $2$) of the results of all such measurements as $u_j$. He measures the second qubit in each pair in the $\Zgate$ basis, and records the XOR of the results of all such measurements as $v_j$. He sets $R_j:=u_j \oplus v_j$ (where $\oplus$ means XOR).
\item He sends the bits $R_j$ ($j=1,\dots,k$) and the bit $w$ to Alice.
\end{enumerate}

\item Alice calculates $y_0:=\sum_{j=1}^k (s_j R_j+t_j w) \mod 2$. Bob calculates $(c+\sum_{j=1}^k u_j) \mod 2$, and sends the resulting bit to Alice. Alice takes the XOR of $y_0$ and the last sent bit from Bob, and the result is the final output.
\end{enumerate}
\end{algorithm*}

Similar to the last section, from the proof of Theorem~\ref{thm4}, a \emph{cheating} Alice learns at most $k + 2$ bits of information about Bob's circuit (i.e. about $a_j$ and the constant $c$). If the last bit of communication is omitted, which may be the case when Scheme~\ref{schsub3} is used in some higher-level schemes such as Scheme~\ref{schmain}, Alice learns at most $k+1$ bits of information about Bob's coefficients $a_j$, and in such case the result is distributed as the XOR of two remote bits. It should be stressed that Alice's learning of the circuit would affect the correctness of evaluation of the distributed result, for similar reasons as discussed after Scheme~\ref{schsub}. If the result is not distributed, but completely on Alice's side, and if $k+2<n$, she still cannot learn all information about Bob's input. If Alice wants to ensure that the evaluation (with distributed result when $k+2\ge n$) is exactly correct, she cannot learn any information about the circuit apart from what is learnable from the output of the evaluation. But she can learn about Bob's input up to the amount mentioned above if she does not care about the correctness of evaluation of the distributed result.

For the data privacy of Scheme~\ref{schsub3}, we consider two criteria: first, when the input is \emph{a priori} uniformly random in the other party's view, the amount of information learnable by Bob is small. The second is the usual IND-CPA type of criterion: that any two classical input strings are hardly distinguishable by Bob. The security under the first criterion is partially discussed in Theorem~\ref{thm7} and in Appendix~\ref{appd3}. It is better than Scheme~\ref{schsub2}, but the precise level of security for large $k$ still needs further study. Later we show that the data security of Scheme~\ref{schsub3} does not satisfy the second criterion, at least when $k\le n$.

\begin{theorem}\label{thm7}
The following statements hold for Scheme~\ref{schsub3} in the case that Alice's input data is uniformly distributed in Bob's view.\\
(i) The information learnable by a cheating Bob about each bit of Alice's input is $\Omega(\frac{1}{2^k})$ bits.\\
(ii) The information learnable by a cheating Bob about the joint distribution of the input bits is at least $\lfloor\frac{n}{2}\rfloor-k$ bits.
\end{theorem}
\begin{proof}
(i) It can be found by explicit calculation that if $k=1$, the trace distance of Bob's two reduced density matrices for a particular $x_i$ is exactly $\frac{1}{\sqrt{2}}$.
For larger $k$, since the encoded states are independent among the qubits with different $j$ (provided that $x_i$ is unknown to Bob), the trace distance of Bob's two reduced density matrices for a particular $x_i$ is exactly $\frac{1}{2^{k/2}}$. Hence, by measuring each qubit in a best basis, the mutual information between Alice and Bob (for uniformly random input $x_i$) by measuring the $k$ qubits is $\Omega(\frac{1}{2^k})$ bits, and numerical calculations suggest that the mutual information and the Holevo bound (limited to one particular $x_i$) may be $\Theta(\frac{1}{2^k})$ bits.

Since the input distribution is uniformly random, each $x_i$ is $1$ with probability exactly $\frac{1}{2}$, then, if Bob had no prior knowledge about the $\{s_j\}$, the $k(n+1)$ qubits received by Bob are in a maximally mixed state in his view, for each possible combination of the values of $s_j$. In other words, Bob's information about the $\{s_j\}$ cannot be increased at all if he had zero information to start with. And he indeed initially has zero information about the $\{s_j\}$. Hence, Bob is completely ignorant of the $\{s_j\}$ even after his measurements. Therefore, his best strategy to learn about one of the bits $x_i$ is to learn it by measuring the corresponding $k$ qubits. As mentioned above, the information about the bit $x_i$ learnable by Bob is $\Omega(\frac{1}{2^k})$ bits for uniformly random input, which may be a tight bound. By doing the same measurement for all input bits, and looking at the correlations among measurement outcomes across different $i$, he may learn about the joint distribution of the input bits.

(ii) Bob may choose to do $\cnot$ gates between pairs of qubits with the same $j$, and measure them in the $\Xgate$ or $\Zgate$ bases as in the scheme. If Alice further sends $k$ bits of information about the $\{s_j\}$ to Bob, and one bit about $\sum_{j=1}^k t_j \mod 2$, he would know all information about the input $\{(x_i+x_{i+1})\mod 2: i\,\,\rm{odd}\,\,\rm{and}\,\,1\le i\le n\}$, which amounts to $\lceil\frac{n}{2}\rceil$ bits. Note that Bob had done a measurement in a fixed bases, so the state he received before measurement is effectively classical. Since locking does not happen for classical states \cite{DHL04}, we get that Bob's information about the $\{x_i\}$ before being told anything about the $\{s_j\}$ and the $\sum_{j=1}^k t_j \mod 2$ is at least $\lceil\frac{n}{2}\rceil-k-1$ bits. And since when $n$ is even, Bob does not need the information about the bit $\sum_{j=1}^k t_j \mod 2$, the information obtainable by Bob in the scheme is actually at least $\lfloor\frac{n}{2}\rfloor-k$ bits.
\end{proof}

Since the quantity in Theorem~\ref{thm7}(ii) is less than that in Theorem~\ref{thm5}, we may regard Scheme~\ref{schsub3} as having better data privacy than Scheme~\ref{schsub2}. But for large $k$, say $k=n$, the amount of information learnable by Bob may be not much smaller than that in Scheme~\ref{schsub2}; further work is needed to give a conclusive statement on this. Note that the security in Scheme~\ref{schsub3} is not always composable, due to the locking property that it relies on: some further small amount of communication from Alice to Bob may reveal much information about her input. This would affect the use of Scheme~\ref{schsub3} in some \emph{cascaded} schemes that aims to achieve interactive QHE, where the data privacy is required to be near-perfect. An example of an ``cascaded scheme'' is that the local evaluation of a linear polynomial near the end of Scheme~\ref{schsub3} is replaced with a bipartite protocol which is an instance of Scheme~\ref{schsub3} (or Scheme~\ref{schsub2} or Scheme~\ref{schsub}), and the initiating party in both the top level and the lower level is Bob rather than Alice. A different cascaded scheme which aims for near-perfect data privacy is described in Sec.~\ref{sec8}. On the other hand, for the bipartite secure computing problems which do not require near-perfect security for either party, the Scheme~\ref{schsub3} would find it use: the simplest example is the computation of the bipartite classical inner-product function, which requires only one use of Scheme~\ref{schsub3}. If we choose $k=1$, Alice would leak at most $\frac{n}{2}$ bits of information according to \cite{DHL04}, and Bob would leak at most $3$ bits of information. This level of security can be regarded as nontrivial because $\frac{n}{2}+3<n$ for large $n$.

The data privacy of Scheme~\ref{schsub3} does not hold under the second criterion, at least when $k\le n$, for the similar reason as in the last section. Suppose the two input strings for Bob to distinguish are the $0\dots000$ versus $0\dots001$. When $k\le n$, he can choose to measure some qubits corresponding to $i<n$, to learn about the $\{s_j\}$, and then use such information and some measurement on the remaining qubit(s) corresponding to $i=n$ to learn about which input string it is with some non-negligible probability. Thus, the Scheme~\ref{schsub2} is not secure under the IND-CPA criterion when $k\le n$.

Although Scheme~\ref{schsub3} is somewhat better than Scheme~\ref{schsub2} in security, it may be harder to implement. In Scheme~\ref{schsub3}, Bob needs to perform $\cnot$ gates on some pairs of qubits determined by his input. This poses problems in experimental implementation, particularly when using photons, which hardly interact with each other. But Bob does not have to send back quantum states to Alice. In comparison, Bob only does $\cnot$ gates on fixed pairs of qubits in Scheme~\ref{schsub2}, but needs to send or teleport back some quantum states to Alice.

\section{Combined use of the previously proposed schemes}\label{sec8}

In this section, we first present the following Scheme~\ref{schsub4} for evaluating classical linear polynomials that makes two uses of Scheme~\ref{schsub3}. Such scheme can be used as a subprocedure in Scheme~\ref{schmain}, giving rise to an interactive QHE scheme with partial data and circuit privacy without resorting to allowing early terminations, but it only works for the case that the data is almost uniformly distributed, and when the circuit is short (see remarks below). Note that by modifying the last steps in the similar way as in previous schemes, the scheme has a version with distributed output (the XOR of two distant bits).

\begin{algorithm*}[htb]
\caption{A scheme for computing classical linear polynomials by two uses of Scheme~\ref{schsub3}}\label{schsub4}
{\bf The type of allowed circuits:} those calculating a linear polynomial $y=(c+\sum_{i=1}^n a_i x_i)\,\,{\rm mod}\,\,2$, where $x_i$ are bit values of Alice's classical input, and $a_i$ and $c$ are constant bits known to Bob.\\
\begin{enumerate}
\item Run a partial instance of Scheme~\ref{schsub3} up to the step of just before Bob sends Alice any message, with $k=\lceil\gamma n\rceil$, where the constant $\gamma<2$ and $\gamma>1$ (e.g. $\gamma=\frac{3}{2}$ is a valid choice).

\item Instead of the local evaluation of a classical linear polynomial $\sum_{j=1}^k s_j R_j +t_j w \mod 2$ in the instance of Scheme~\ref{schsub2} above, run a bipartite protocol: it is almost an instance of Scheme~\ref{schsub3} with $k'=1$  (we use $k'$ here to distinguish from the $k$ in Step 1), with Bob as the first party and Alice as the second party, but with the last sending of the mask bit omitted.

\item Alice and Bob perform the remaining operations in the instance of Scheme~\ref{schsub3}. The distributed outcome in Step 2 is treated by Bob sending Alice his bit in the distributed pair of bits.
\end{enumerate}
\end{algorithm*}

A brief security analysis of the scheme is as follows. The information about Alice's input $\{x_i\}$ leaked to Bob in Step 1 is $O(\frac{{\rm poly}(n)}{2^{k-n}})$ bits, due to that $k=\gamma n$ with $\gamma$ some constant between $1$ and $2$, and that the data privacy in Scheme~\ref{schsub3} can be regarded as being similar to but slightly better than Scheme~\ref{schsub2}, and the latter is analyzed in Appendix~\ref{appd3}. Such amount of leakage of the data is exponentially small. Alice's sending of up to $2$ bits in Step 2 may reveal some information about the $s_j$ (and hence about the $x_{ij}$ in the higher level scheme) in the instance of Scheme~\ref{schsub}, but since $k=\gamma n$ and that Alice's inputs in Step 2 are the $s_j$, no definitive information about $\{x_i\}$ is leaked, and when $n$ is large, the amount of leaked information about $\{x_i\}$ is quite small, even if Bob cheats. For $\gamma=\frac{3}{2}$, we estimate that a cheating Bob may learn an exponentially small amount of information in total about Alice's input. On the other hand, Bob may leak up to half of the input size in Step 2, i.e. at most $\lceil\frac{\gamma n}{2}\rceil$ bits. Thus, about $(1-\frac{\gamma}{2})n$ bits of information about Bob's coefficients $a_i$ is secure from a cheating Alice.

If the use of Scheme~\ref{schsub3} in the Steps 1 and 3 is replaced with Scheme~\ref{schsub2}, the security is harmed in the case that Alice is dishonest: she could use some entangled states on pairs qubits (e.g. same as that mentioned in the analysis of Scheme~\ref{schsub}) to find out Bob's input from the returned state from Bob, although this would make the computation result incorrect. In contrast, if  Scheme~\ref{schsub3} is used in the top level as described above, Bob does not send back any quantum state except in the Step 2, so this makes the input of Bob much safer.  Interestingly, there is a classical analog of Scheme~\ref{schsub2} (which can also be regarded as a classical analog of Scheme~\ref{schsub3}), such that when it is used in the Steps 1 and 3 in Scheme~\ref{schsub4} in place of Scheme~\ref{schsub3}, the overall circuit privacy is about the same, while the data privacy weakens slightly; the overall scheme becomes easier to implement since a large part becomes classical. The description of such classical scheme is listed as Scheme~\ref{schsubclassical} below. Note that by modifying the last steps in the similar way as in previous schemes, the scheme has a version with distributed output (the XOR of two distant bits).

\begin{algorithm*}[htb]
\caption{A classical scheme for computing classical linear polynomials}\label{schsubclassical}
{\bf The type of allowed circuits:} those calculating a linear polynomial $y=(c+\sum_{i=1}^n a_i x_i)\,\,{\rm mod}\,\,2$, where $x_i$ are bit values of Alice's classical input, and $a_i$ and $c$ are constant bits known to Bob.\\
\begin{enumerate}
\item Alice and Bob agree on a positive integer $k$ related to the security of the scheme. The suitable choice of $k$ is $O(n)$ in general, but may be chosen to be $O(\log n)$ if correlations among input bits may be leaked.
\item
\begin{enumerate}
\item For each $i$, Alice randomly chooses $k$ bits $x_{ij}$ satisfying that $x_i=\sum_{j=1}^k x_{ij} \mod 2$. She generates independent unbiased random bits $s_j$, where $j=1,\dots,k$. The $s_j$ are independent from $x_{ij}$.
\item For each $(i,j)$ pair, Alice prepares two bits. When $s_j=0$, the first bit in the pair is equal to $x_{ij}$, otherwise the second bit in the pair is $x_{ij}$. The other bit in the pair is a completely random bit.
\item Alice sends the resulting $kn$ bits to Bob.
\end{enumerate}

\item
For each $j\in\{1,\dots,k\}$,Bob records the sum modulo $2$ of the first bits in those pairs corresponding to $a_i=1$ with the given $j$ as $u_j$, and the similar quantity for the second bits in the pairs as $v_j$. He sets $R_j:=u_j \oplus v_j$ (where $\oplus$ means XOR). He sends the bits $R_j$ ($j=1,\dots,k$) to Alice.

\item Alice calculates $y_0:=\sum_{j=1}^k s_j R_j \mod 2$. Bob calculates $(c+\sum_{j=1}^k u_j) \mod 2$, and sends the resulting bit to Alice. Alice takes the XOR of $y_0$ and the last sent bit from Bob, and the result is the final output.
\end{enumerate}
\end{algorithm*}

The analysis of the data privacy in Scheme~\ref{schsubclassical} is similar to that of Scheme~\ref{schsub2}, so the data privacy is somewhat weaker than Scheme~\ref{schsub3} under the same parameter $k$. The circuit leakage is upper bounded by the amount of communication. The use of Scheme~\ref{schsubclassical} alone is not very interesting, but the combined use with Scheme~\ref{schsub3} as mentioned above makes it quite interesting. Another possibility is to combine it with classical homomorphic encryption schemes, which is out of the scope of the current work.

Another potential alternative to Scheme~\ref{schsub4} is using Scheme~\ref{schsub3} alone with some suitable choice of $k$, but this method does not enjoy the security advantage from the backward instance of Scheme~\ref{schsub3} in Step 2 above, thus we expect that it cannot fully reproduce the security characteristics of Scheme~\ref{schsub4}.

In the Step 2 of Scheme~\ref{schsub4}, if an instance of Scheme~\ref{schsub3} with $k'>1$ is used, the amount of data leakage is somewhat larger, but the circuit privacy improves a lot. A pessimistic estimate of the amount of information leakage of Bob's circuit is about $\frac{\gamma}{2}n-\frac{k'}{2}+C$ bits, where $C$ is a constant (this is by borrowing from the estimate for Scheme~\ref{schsub2}, since it is believed that the security of Scheme~\ref{schsub3} is not worse than Scheme~\ref{schsub2}), and since $k'$ should be slightly smaller than $(\gamma-1)n$ to make the data secure (c.f. Appendix~\ref{appd3}), we estimate that the information leakage of Bob's circuit is not more than $(\frac{1}{2}+\e)n+C$ bits, where $\e$ is a small positive constant. In other words, $(\frac{1}{2}-\e)n-C$ bits of information about Bob's circuit can be secure.

By using the scheme as a subprocedure in Scheme~\ref{schmain} (which requires slight modifications to the scheme to make the output distributed, i.e. represented by the XOR of two distant bits), we obtain an interactive QHE scheme with partial data and circuit privacy without resorting to allowing early terminations. The data privacy is quite good for short circuits, since each use of the subprocedure by itself leaks an exponentially small amount of information to Bob, while the repeated uses of the subprocedure may give rise to worse data privacy. Because of the construction of Scheme~\ref{schsub3}, the leaked information in one subprocedure may help the cheating Bob significantly in the next instance of the subprocedure, because the information from the previous subprocedure helps Bob learn the $s_j$ in the next subprocedure, which in turn helps him learn some of the variables $x_i$. Similar to other schemes mentioned previously, Alice's or Bob's cheating necessarily affects the correctness of evaluation. Given the above, without further using the verification techniques, the obtained interactive QHE scheme is only good for the case that the number of $\tgate$ gates is $O(n)$, and only for the case that the data is almost uniformly distributed. Because of the limitation to circuit size, the circuit privacy is also not ideal: Bob may hide about $O(n^2)$ bits of information about the circuit. However, by combining with the verification techniques (such as that in Scheme~\ref{schcheck}) in an interactive protocol, better data and circuit privacy may be obtained.

As a way to avoid the requirement that ``the data is almost uniformly distributed'', it is possible to use a different overall structure of QHE than Scheme~\ref{schmain}, namely, keeping data on Alice's side. The techniques can be partially seen in the first version of the current draft and also in some blind quantum computing schemes. We put such method and related experiment in a separate work. Such method still has the similar limitation on the circuit depth (e.g. limitation on the number of $\tgate$ gates).

A lesson learned from the above analysis is that repeated evaluation of linear polynomials with the same variables would often affect data privacy significantly. To avoid this problem, in Sec.~\ref{sec9} below, we consider letting Alice use newly generated intermediate variables as much as possible in the linear polynomials. (Two other protocols based on the similar idea are also considered.) This seems a natural fit for the problem of evaluation of classical functions, rather than (interactive) QHE.

\section{Applications in secure evaluation of classical functions}\label{sec9}

In this section, we consider the problem in which Alice tries to evaluate a classical function known to Bob on her private data, and the ideal goal is to keep the data private from Bob, and the function private from Alice, apart from what can be learnt from the output. The actual security is worse, and there is some limitation to the circuit size. The steps of the scheme are as follows. Note that no initial quantum-one-pad nor its classical counterpart are used.\\

1. Decompose the classical circuit to be evaluated into a form which consists of evaluation of classical linear polynomials and the XOR and AND functions of some bits, where the input and output of such functions may be intermediate variables.

2. For each classical linear polynomial, use the combined scheme for evaluating classical linear polynomials introduced in the last section, but with the sending by Bob of the last bit omitted. Hence, the result is distributed across two parties. Rewrite the later XOR and AND functions and the linear polynomials in terms of the distributed intermediate variables.

3. For the next function in the circuit, continue to do Step 2, until the end of the circuit is reached.\\

\smallskip
Note that the AND functions above refers to the AND between two of Alice's bits, so it is not a special case of a linear polynomial with bipartite input. The data privacy of the scheme depends very much on how many times the (original or intermediate) variables are reused. The more they are reused, the worse the data privacy is. If each variable is reused for at most a constant number of times, then the data privacy is quite good: Bob learns a negligible amount about the variables, unless he chooses to cheat by using trivial programs, so that the intermediate variables carry the equivalent information about the original variables. Hence, the correctness of the scheme implies the data privacy. Since the circuit can be polynomial in the input size $n$, Bob can hide $O(nm)$ bits, were $m$ is the number of linear polynomials in the decomposed circuit, and it is assumed that all linear polynomials have input size $O(n)$ and at least a constant fraction of them are of size $\Theta(n)$. This protocol still has the problem that many intermediate results may depend on the same set of initial variables, hence the initial variables are not very secure for long circuits. Thus, the number of allowed linear polynomials in the circuit is restricted to $O(n)$ and less than $c\cdot n$, where $c$ is some constant less than $1$ determined from the security analysis of Scheme~\ref{schsub2}.

To resolve this problem, we may change the combined scheme for evaluating linear polynomials in Step 2, so that at the top level it uses Scheme~\ref{schsub} rather than Scheme~\ref{schsub2}, in order to achieve better data privacy. This gives rise to worse circuit privacy, but from the fact that Alice cannot learn perfectly about both Bob's input and the masks, Alice's cheating would affect the correctness of evaluation (of the distributed result) after many gates. She can choose to always cheat by learning Bob's coefficients without learning Bob's masks at each polynomial, but then the final result would be on Alice's side only and is not distributed, so if some bits in the final result are to be sent to Bob, he would acquire a wrong result after taking the XOR of the sent bits from Alice and his local mask bits. Thus, we see that the altered method in this paragraph is only good for a changed problem in which the outcome of the computation is at least partially on Bob's side. Bob may use some output bits on his side to check that Alice has cheated or not. Such changed problem is a type of general bipartite computing problem, while the original problem is a more restrictive one.

The protocol in the previous paragraph hints at a different approach to the original problem: let Bob act as the first party in the instances of Scheme~\ref{schsub2} for  evaluating the linear polynomials, and treat his classical program as his data; and Alice's data is treated as her program. Then, since the final result is on Alice's side, this case matches the situation in the previous paragraph. Alice's input is not too secure after many evaluations of the linear polynomials. But in the case that the output is only one bit, we may postpone the local evaluation of linear polynomials to the end of the overall protocol, to do a one-level locking scheme (Scheme~\ref{schsub3} with $k=1$, but with the sending of the mask bit omitted) at the end, so that Alice only leaks a constant number of bits ($2$ bits) to Bob, while Bob's circuit is not completely known to Alice. Finally Bob sends Alice a mask bit for her to obtain the final outcome. This is very different from the security characteristics obtained in the protocols in the previous paragraphs, in which Alice's data is asymptotically secure. The data privacy is partial, similar to that obtainable by Scheme~\ref{schcheck}, but note that success is guaranteed here, while Scheme~\ref{schcheck} involves early terminations in the case that some party cheats; but the output size here is limited to one bit (or a constant number of bits if we allow Alice to leak a constant number of bits of information about her data), while Scheme~\ref{schcheck} allows the output size to be comparable to input size. The protocol here also assumes that Bob's circuit is quite random within the possible choices, which is a reasonable assumption.

For the interactive QHE problem, the methods above do not directly apply, since quantum information cannot be copied in general.

\section{Conclusions}\label{sec10}

In this paper, we have constructed a quantum homomorphic encryption scheme for a restricted type of circuits with perfect data privacy for real product input states. We then constructed a QHE scheme for a larger class of polynomial-depth quantum circuits, which has partial data privacy. Both schemes have good circuit privacy. The entanglement and classical communication costs scale linearly with the product of input size and circuit depth. We then introduced the interactive Scheme~\ref{schmain} for general polynomial-sized quantum circuits, which has asymptotic data privacy, but the circuit privacy is not good. We proposed Scheme~\ref{schcheck} for improving circuit privacy by embedded verifications, while the data privacy is good except for a leak of a constant number of bits. It works for the case that the input size is larger than some constant determined by the security requirements. Finally, we present some schemes for computing classical linear polynomials which achieve some nontrivial level of data privacy and circuit privacy, and show applications in partially secure evaluation of classical functions and some special case of the QHE problem.

The reason why rebits are used in Schemes~\ref{sch1} and \ref{sch2} here is to avoid data leakage: if qubits with complex amplitudes are used, after some measurements are done on the Alice side in the protocol, and the outcomes are sent to Bob, it is often the case that some information about the data would have been sent to Bob. It would be good to understand better the ``phase qubit'' and its possible applications. It would also be good to combine the techniques in Schemes~\ref{sch1} and \ref{sch2} with other schemes to yield some interesting protocol, not necessary for the QHE problem.

Some further issues to investigate include: a better characterization of the type of circuits allowed in Scheme~\ref{sch2}; better characterization of the security of the schemes for evaluation of classical linear polynomials; how to optimize the schemes for particular sets of allowed circuits; how to turn Scheme~\ref{schmain} into a non-interactive scheme (possibly restricting the type of circuits or with compromise on security); how to make the schemes fault-tolerant; whether there are corresponding schemes in the measurement-based computing model. In view of the results on classical-client blind quantum computing by Mantri \emph{et al} \cite{Mantri17}, and computationally-secure QHE scheme with almost classical client \cite{Mahadev17}, we may ask whether there is a QHE scheme with some nontrivial level of information-theoretic data privacy and circuit privacy for polynomial-sized circuits, satisfying that one of the parties is fully classical.

\smallskip
\section*{Acknowledgments}

This research is funded in part by the Ministry of Science and Technology of China under Grant No. 2016YFA0301802, and the Department of Education of Zhejiang province under Grant No. Y201737289, and the startup grant of Hangzhou Normal University.

\linespread{1.0}
\bibliographystyle{unsrt}
\bibliography{homo}

\begin{appendix}
\section{Proof of data privacy in Schemes~\ref{sch1} and \ref{sch2}}\label{appd1}

\noindent\textbf{(1). Proof for Theorem~\ref{thm1}.}
\bpf
Alice's measurement outcomes in the gate gadgets are uniformly random and independent. In the following we argue that Alice's fixed gates, gate gadgets together with  sending of her measurement outcomes to Bob do not reveal information about the data. As a preparation, since Alice's measurement outcomes can be compensated by $\Zgate$ gates on Bob's qubits, we may remove Alice's qubit in the EPR pair in the gate gadgets and directly link Bob's qubit in the gadget with Alice's data qubit with a controlled-$i\sigma_y$ gate.

After the above preparation, the fixed gates and the gate gadgets on Alice's side contain only Clifford gates. This means that if we apply some initial $\R_y(\pi)$ operator on Alice's input qubits, they will ``commute'' through the circuit with Pauli corrections on all the qubits, including Bob's qubits. Since Bob's qubits only interact with Alice's qubits via the controlled-$i\sigma_y$ gates, it can be found that any nontrivial Pauli correction on Bob's qubits can only be $\Zgate$ up to a phase. For each input data qubit, there is a subset of Bob's qubits that is subject to $\Zgate$ corrections resulting from the $\Xgate$ operator on the associated Alice's qubit, which is due to initial $\R_y(\pi)$ on that input data qubit. An initial $\R_y(\pi)$  on the phase qubit does not change Bob's reduced density operator. Since Bob's qubits are connected to Alice's part of the circuit via the controlled-$i\sigma_y$ gates, the $Z$-basis information of each Bob's qubit is reflected in Alice's part of the circuit as a choice of $\igate$ or $\R_y(\pi)$. Let us consider the reduced density operator of Bob's. We argue that the off-diagonal terms are such that they are insensitive to the $\Zgate$ gates induced by Alice's initial $\R_y(\pi)$ masks (the diagonal terms are trivially so, since $\Zgate$ only applies phases on the computational-basis states). Consider an off-diagonal term of the form $\ketbra{(g_i)}{(h_i)}$, where $\ket{(g_i)}$ is the short-hand notation for $\ket{g_1}\ket{g_2}\dots\ket{g_m}$, where $g_1,g_2,\dots,g_m$ are bits, and $m$ is the total number of Bob's qubits. Note that the $\Zgate$ gates induced by a fixed initial mask pattern [$\R_y(\pi)$ on some input data qubits] may only bring a sign to this off-diagonal term when
\begin{eqnarray}\label{eq:gihiprime}
\sum_{i=1}^m \left(g'_i+h'_i\right)\,\,\rm{mod}\,\,2=1,
\end{eqnarray}
where $g'_i=g_i$ when the qubit $i$ is subject to the $\Zgate$ correction resulting from this mask pattern, and $g'_i=0$ otherwise; the $h'_i$ is similarly defined.

The two bit strings $(g_i)$ and $(h_i)$ correspond to two patterns of $\R_y(\pi)$'s in Alice's part of the circuit. In order to calculate the coefficient for the off-diagonal term $\ketbra{(g_i)}{(h_i)}$, we need to calculate the partial trace on Alice's side. Consider using the circuit diagram of Alice's side with the target qubits in controlled-$i\sigma_y$ gates replaced with $\R_y(\pi)$ when the corresponding $g_i=1$, and another circuit diagram similarly obtained using the bits $h_i$ instead of the $g_i$. The partial trace is usually calculated by using a concatenated circuit with the output ends of the two original circuits connected on each qubit line. Since it is symmetric, we may instead consider a ``difference'' circuit where the $\R_y(\pi)$ appears at the qubit originally linked to Bob's $i$-th qubit when
\begin{eqnarray}\label{eq:gihi}
\left(g_i+h_i\right)\,\,\rm{mod}\,\,2=1.
\end{eqnarray}

Suppose some initial mask pattern changes an off-diagonal term of the form $\ketbra{(g_i)}{(h_i)}$ in Bob's reduced density operator, then it must be that this term does not vanish. We now prove that the term satisfying Eq.~\eqref{eq:gihiprime} must vanish, then we will get that no initial mask pattern may change an off-diagonal term. Note that when we commute $\Xgate$ or $\R_y(\pi)$ gates through a data qubit line in the circuit, they either preserve themselves (with possible corrections on neighboring qubits) or swap with each other. Thus, only an even number of $\R_y(\pi)$ operators on a qubit line may cancel out after commuting them on this line. If there is a qubit line on which the $\R_y(\pi)$ operators resulting from nonzero $g_i$ and $h_i$ cancel out [it must be that their resulting $\R_y(\pi)$ operators on the phase-qubit line also cancel out], then we may remove these operators, and change these $g_i$ and $h_i$ (and associated $g'_i$ and $h'_i$, if any) to zero, without affecting Eqs.~\eqref{eq:gihiprime} and \eqref{eq:gihi} (which is due to the definition of $g'_i$ and $h'_i$ and the types of gates in the circuit), and look at a reduced case with fewer nonzero $g_i$ and $h_i$. After these reductions, we have that on at least one qubit line, the operators do not cancel out. [If no such line exists, Eq.~\ref{eq:gihiprime} is violated.] We may cancel out the $\R_y(\pi)$ operators on the same line in pairs, without affecting the two equations, so that only one or two operators remain on each line. On each data qubit other than the first data qubit, either there is no correction operator, in which case this line can be removed without affecting the two equations, or there is only one $\R_y(\pi)$ operator, which when applied on the real input state on that data qubit would give rise to zero trace on Alice's side. Now we consider the operators on the first data qubit and the phase qubit. Since we only consider those off-diagonal terms satisfying Eq.~\eqref{eq:gihiprime}, the joint operator on these two qubits (corresponding to such off-diagonal term) must be one of the following: $\Xgate\ox\R_y(\pi)$, or $\Zgate\ox\R_y(\pi)$. Since the initial state on these two qubits are real, it can be found that these operators map the initial state to an orthogonal state. Therefore the partial trace on Alice side is zero, which implies that the term satisfying Eq.~\eqref{eq:gihiprime} must vanish. Thus it must be that no initial mask pattern may change any off-diagonal term in Bob's reduced density operator.

In the above, we have shown that for any pure input state of the required type, and its $\R_y(\pi)$-padded states [``$\R_y(\pi)$-padded'' means that there are $\R_y(\pi)$ gates applied on some of the qubits, possibly including the phase qubit], the reduced density operators of Bob's system after correcting for Alice's measurement outcomes are the same. Thus, for a subset of such input states, namely real product states, Bob's reduced density operator is equal to that for the input being the average state of all $\R_y(\pi)$-padded states, which is a fixed maximally-mixed state. Thus, these reduced density operators of Bob's are fixed regardless of what the pure input state is. Thus Bob does not receive any information about the input data in the case that the input state is restricted to be a real product state. For the other type of input pure state, that is, the tensor product of an arbitrary two-qubit real state on the first data qubit and the phase qubit, and a real product state of other data qubits, we can similarly consider applying initial $\Zgate$ or $\Xgate$ masks to the first data qubit of the input state, and $\R_y(\pi)$ masks on other data qubits, and similarly find that such initial mask pattern does not change any off-diagonal term in Bob's reduced density operator.  By considering initial Pauli masks on the phase qubit, we find that all possible Pauli masks on the phase qubit does not change Bob's reduced density operator. Consequently, all possible Pauli masks on the first data qubit and the phase qubit does not change Bob's reduced density operator. Therefore, the logical input state of the first qubit (encoded on the first data qubit and the phase qubit) does not affect Bob's reduced density operator. The $\R_y(\pi)$ masks on other data qubits still do not matter. This implies that Bob receives no information about the input, for pure input states of the required type. The case of probabilistic mixture of such states follows by linearity. This completes the proof.
\epf

\noindent\textbf{(2). Proof for Theorem~\ref{thm2}.}
\begin{proof}
(i) As in the proof of Theorem~\ref{thm1}, to consider a change from one of the two input states to the other, we only need to consider applying $\R_y(\pi)$ on each data qubit and a $\R_y(\theta)$ on the phase qubit. These operators commute through the circuit, and they sometimes become $\Xgate$ operators, resulting in some $Z$ operators on some of Bob's qubits. Consider an off-diagonal term in Bob's reduced density matrix. Such term does not vanish if and only if some $\R_y(\pi)$ operators on some data-qubit lines in some ``difference'' circuit commute back to the beginning of the circuit and meet with the input state to yield nonzero trace on Alice's side.
The operators at the beginning of the circuit is always a product operator, and if $\R_y(\pi)$ appears on the phase qubit, the type of input state in the assertion to be proved would be mapped to an orthogonal state, giving rise to zero trace on Alice's side. The case of nonzero trace only occurs when there are an even number of $\R_y(\pi)$ in the ``difference'' circuit that are in positions that were subject to $\Xgate$ operators in the original circuit (each resulting in a $\Zgate$ operator on one of Bob's qubits). From Eq.~\ref{eq:gihiprime}, this means that Bob's off-diagonal term does not vanish only when it is not affected by the choice of input state (between the two possible input states). Thus, Bob's reduced density matrix is independent of the choice of input state (between the two possible input states).\\
\indent (ii) The argument follows from the proof of Theorem~\ref{thm1}, including the ``averaging'' argument in the last paragraph of the proof of Theorem~\ref{thm1}. The one-qubit Pauli operators on the data qubit in the proof of Theorem~\ref{thm1} are replaced with operators on some data qubits here. In particular, the initial logical $\R_y(\pi)$ on $n$ data qubits (where $n$ is odd) is regarded as $\R_y(\pi)$ on all data qubits. The initial logical $\Zgate$ on $n$ data qubits is regarded as $\Zgate$ on all data qubits. The assumption of $n$ being odd is essential, since $\R_y(\pi)$ on an even number of qubits is not the logical $\R_y(\pi)$, under the current encoding of $\ket{0}$ and $\ket{1}$ into $\ket{0}^{\ox n}$ and $\ket{1}^{\ox n}$, respectively. It follows from the above types of correspondences of Pauli operators that the first $n$ qubits of the two encoded states could also be any computational-basis state and $\Xgate^{\ox n}$ acting on such state. This completes the proof.
\end{proof}

\section{The garden-hose gadget that corrects an unwanted $\pgate$ gate}\label{appd2}

The Fig.~\ref{fig:toy2} below shows a simplified version of a gadget in \cite{Dulek16} for correcting an unwanted $\pgate$ gate due to a $\tgate$ gate in the circuit with certain prior Pauli corrections. The input qubit starts from the position ``in'', and ends up in a qubit on Bob's side labeled ``out1'' or ``out2'', depending on Bob's input bit $p$. The unwanted $\pgate$ on this qubit is corrected, but some other Pauli corrections are now needed because of the Bell-state measurements. These Pauli corrections are to be accounted for in the later evaluation of linear polynomials in the instances of Scheme~\ref{schsub} in the main Scheme~\ref{schmain}.
Note that in each use of this gadget, some of the Bell-state measurements are not actually performed, dependent on the value of $p$ and $q$; and that Alice's Bell-state measurements are on the same pairs of qubits irrespective of $q$.

\begin{figure*}[h!]
\centering
\begin{tikzpicture}[decoration=snake]
\filldraw (0,0) circle (2pt);
\filldraw (2,0) circle (2pt);
\draw[decorate] (0,0) -- (2,0);

\filldraw (0,-1) circle (2pt);
\filldraw (2,-1) circle (2pt);
\draw[decorate] (0,-1) -- (2,-1);

\filldraw (0,1) circle (2pt);
\filldraw (2,1) circle (2pt);
\draw[decorate] (0,1) -- (2,1);

\filldraw (0,2) circle (2pt);
\filldraw (2,2) circle (2pt);
\draw[decorate] (0,2) -- (2,2);

\node at (-1.5,4) {$p$};
\node at (-3,3.5) {$0$};
\node at (-0.3,3.5) {$1$};
\draw (-0.3,1) to [bend left=50] (-0.3,3);
\draw (-3,2) to [bend left=50] (-3,3);

\node at (5,4) {$q$};
\node at (3.3,3.5) {$0$};
\node at (6.5,3.5) {$1$};
\draw (3.3,0) to [bend right=50] (3.3,2);
\draw (6.5,-1) to [bend right=50] (6.5,1);

\filldraw (0,3) circle (2pt);
\node[anchor=west] at (0,3) {in};

\node[anchor=west] at (-0.2,0.3) {out1};
\node[anchor=west] at (-0.2,-0.7) {out2};

\draw (3.3,-1) to [bend right=50] (3.3,1) to (2.3,1);
\draw (6.5,0) to [bend right=50] (6.5,2) to (5.5,2);
\filldraw[fill=white] (5.7,1.75) rectangle (6.3,2.25);
\filldraw[fill=white] (2.5,0.75) rectangle (3.1, 1.25);
\node at (6,2) {$\pgate^\dag$};
\node at (2.8,1) {$\pgate^\dag$};

\draw[dashed] (-4.5,-2) rectangle (0.75,5);
\draw[dashed] (1.3,-2) rectangle (8,5);
\node[anchor=south] at (-4,5) {Bob};
\node[anchor=south] at (7.5,5) {Alice};
\end{tikzpicture}
\caption{A simplified version of a gadget in \cite{Dulek16} for applying a $\pgate^\dag$ to a qubit initially at the position ``in'' if and only if $p+q=1\,\,(\rm{mod}\,\,2)$, using the ``garden hose'' method. The dots connected by wavy lines are EPR pairs. The curved lines are for Bell-state measurements. For example, if $p=0$ and $q=1$, the qubit is teleported through the first and the third EPR pairs, with a $\pgate^\dag$ applied to it by Alice in between. The state of the input qubit always ends up in a qubit on Bob's side, and the position depends on Bob's input bit $p$: if $p=0$, the output is on the qubit labeled ``out1'', otherwise it is on the qubit labeled ``out2''. Note that in each use of this gadget, only $3$ Bell-state measurements are actually performed; Alice's Bell-state measurements are on the same pairs qubits irrespective of $q$.}
\label{fig:toy2}
\end{figure*}
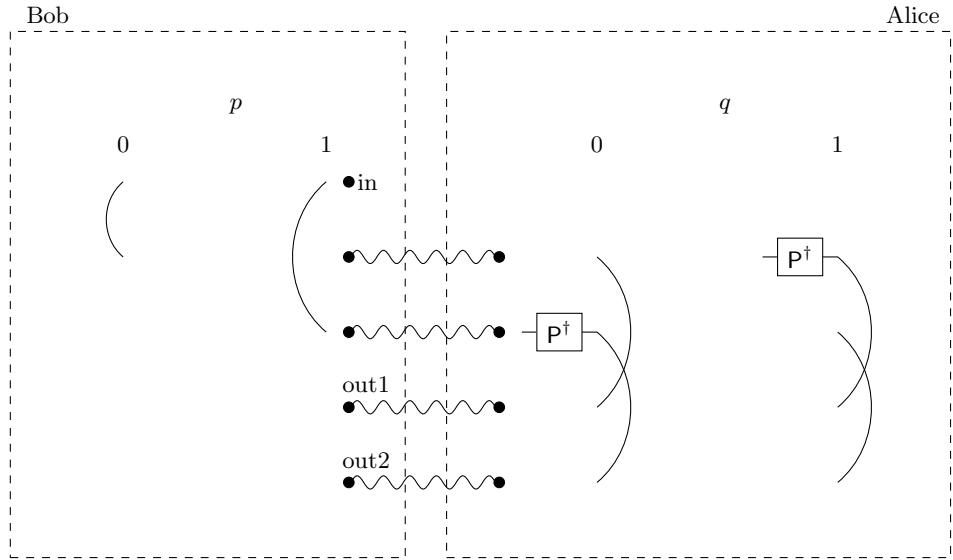

\section{Classical mutual information in the case of uniformly distributed data in Scheme~\ref{schsub2}}\label{appd3}

This appendix complements Theorem~\ref{thm5}. We consider the classical mutual information under the condition that Alice's input distribution is uniformly random. Such quantity is abbreviated as CMI below.

First, we show that in the case $k=1$ in Scheme~\ref{schsub2}, the CMI between Alice's input and Bob received state is $n-1+\frac{1}{2^n}$ bits.

The construction of the scheme implies that Bob should always measure the first qubit in each pair in the $Z$ basis, and the second qubit in each pair in the $X$ basis, to get maximal information about the data. In the following, we assume that Bob indeed does such measurements, and regard the outcomes as the classical bits obtained by Bob. Since the input and output are both effectively classical, the classical mutual information is equal to the Holevo quantity. Thus, a direct calculation of the Holevo quantity should give the CMI. In the following we give an alternative way of calculating the CMI in the case of $k=1$ with uniformly distributed data.

Since the input distribution is uniformly random, the bits received by Bob are uniformly random. For each input bit $x_i$, there is one half probability that the corresponding two bits obtained by Bob is $00$ or $11$. In such case, the bit $x_i$ is determined. In other cases, two input bits $x_i$ and $x_j$ may be correlated. This can happen when the two bits obtained by Bob are $01$ or $10$, for both $x_i$ and $x_j$. The contribution to the overall mutual information is $\frac{1}{4}$ bits, because the case just mentioned happens with probability $\frac{1}{4}$. If we add up all $C_n^2$ such pairs, there is some overcounting. We may subtract the three-pair coincidences, then add back the four-pair coincidences, etc, and get the following expression for the classical mutual information:

\bea\label{eq:mutualinfo}
&&\frac{n}{2}+\sum_{j=2}^n \frac{(-1)^j}{2^j}C_n^j\notag\\
&=&\frac{n}{2}+(1-\frac{1}{2})^n-1+\frac{n}{2}=n-1+\frac{1}{2^n}\quad\rm{(bits)}
\eea

In the following we discuss the general case of $k>1$ in Scheme~\ref{schsub2}. Let us first look at the special case of $n=k=2$. For each input bit $x_i$, there is one half probability that the corresponding two bits obtained by Bob for each qubit is $01$ or $10$. For the input bit $x_i$, the probability that both pairs of bits obtained by Bob are $00$ or $11$ is $\frac{1}{4}$, which contributes $\frac{1}{4}$ bits of information to the CMI, and two input bits contribute $\frac{1}{2}$ bits in total. For each input bit $x_i$, the probability that both pairs of bits obtained by Bob are $01$ or $10$ is $\frac{1}{4}$. If both the bits for $x_1$ and $x_2$ satisfy this, Bob gets $1$ bit of information about the correlation between the two input bits. This gives $\frac{1}{16}$ bits of contribution to the CMI. The same amount of correlation information may also be obtained by Bob when the two pairs of bits for $x_1$ obtained by Bob contain exactly one pair being $01$ or $10$, and exactly the same happens for $x_2$ with the $01$ or $10$ in the same position labeled by $j$ in Scheme~\ref{schsub2}. Such cases give $2\times\frac{1}{16}$ bits of contribution to the CMI. The overall CMI in the case $n=k=2$ is $\frac{1}{2}+3\times(\frac{1}{4})^2=\frac{11}{16}$ bits, where $3=2^2-1$, with the explanation being that the case of both input bits correspond to $00$ or $11$ at every position is excluded. By the similar argument, For general $k$ and fixed $n=2$,
the CMI is $\frac{2}{2^k}+\frac{2^k-1}{2^{2k}}=\frac{3}{2^k}-\frac{1}{2^{2k}}$ bits.

For $n>2$ and general $k$, we may combine the methods for the case of $n>2$ and $k=1$, and the case of $n=2$ and general $k$, and obtain that the part of CMI induced by the single-bit and two-bit correlations among $\{x_i\}$ is
\bea\label{eq:mutualinfo2}
&&\frac{n}{2^k}+\sum_{j=2}^n (-1)^j\frac{2^k-1}{2^{jk}}C_n^j\notag\\
&=&\frac{n}{2^k}+(2^k-1)[(1-\frac{1}{2^k})^n-1+\frac{n}{2^k}]\notag\\
&=&n-(2^k-1)[1-(1-\frac{1}{2^k})^n]\quad\rm{(bits)}.
\eea
When $n$ and $k$ are large, and $k>C\log_2 n$ for some constant $C>2$, the expression above is $O(\frac{n^2}{2^k})$ bits. But multi-bit correlations are ignored above. To consider the effect of multi-bit correlations, note that the linear correlations among some subset of $\{x_i\}$ is learnable by Bob if and only if the corresponding vectors are linearly dependent, where the vector for a $x_i$ is defined as a length-$k$ bit string, and each nonzero bit in such string represents that some pair of bits obtained by Bob is $01$ or $10$ (as opposed to $00$ or $11$). When $n=k$ and both are large, the $n$ bit strings contain, on average, about a constant number of linearly dependent strings.  Hence, we estimate that the CMI is about a constant number of bits when $n=k$ and both are large. For the case that $k$ is much larger than $n$, say when $k=\lceil\gamma n\rceil$ with the constant $\gamma>\frac{4}{3}$, we estimate that the number of linearly dependent strings are zero, thus the CMI is near zero. The CMI in this case is $O(\frac{{\rm poly}(n)}{2^{k-n}})$ bits, as the contribution from multi-bit correlations is not negligible compared to that of two-bit correlations. But the circuit privacy is quite bad in this case, since $k>n$; such parameter range may still be interesting due to the possible methods of verifications by Alice, or in the case discussed in Sec.~\ref{sec8}.
\end{appendix}
\end{document}